%% file: main.tex
\documentclass[11pt]{article}
\usepackage{fullpage}
\usepackage{epsfig}
\usepackage{multienum,subfigure}
\usepackage{algorithm}
\usepackage[noend]{algpseudocode}
\usepackage{amsmath,amssymb,amsthm}
\usepackage{ifsym}
\usepackage{verbatim}
\usepackage[dvipsnames]{xcolor}
\usepackage{graphicx}
\usepackage{bbm}
\usepackage{todonotes}
\usepackage{xspace}
\usepackage[normalem]{ulem}

\usepackage{mathtools}
\usepackage{enumitem,needspace,booktabs}
\usepackage{thm-restate}
\usepackage[square,numbers]{natbib}
\usepackage{hyperref}
\usepackage[nameinlink,noabbrev]{cleveref}

\input{paper-macros}

\hypersetup{pdftitle={Utility maximization with i.i.d.~values: pricing, free allocation, and uniform rationing},pdfauthor={Research lecture notes}}
\title{A Constant-Factor Approximation to Multidimensional \\ Consumer Utility}
\author{Kira Goldner \and Taylor Lundy \and Thodoris Tsilivis}
\date{\today}
\begin{document}
\maketitle
\thispagestyle{empty}

\begin{abstract}
    Motivated by social services where consumers pay with non-transferable ordeals, we study mechanisms that maximize consumer utility for multiple unit-demand buyers and heterogeneous items whose values are drawn independently from known prior distributions. Prior work in utility maximization approximates social welfare and shows that the gap between optimal utility and social welfare is logarithmic. We resolve the question of whether simple mechanisms can guarantee a constant-factor approximation to optimal utility itself. Our mechanisms achieve a $(5.67+\varepsilon)$-approximation for general independent values, improving to $2e/(e-1)<3.164$ when values are i.i.d. Each buyer chooses their favorite option from posted item prices or free item lotteries, and contention resolution determines which buyers' requests are served; this is BIC, ex-post individually rational, and computable in polynomial time. Our main technical contribution is a general upper bound on the optimal ex-ante-constrained utility that separates the contributions captured by posted prices and free lotteries. These results establish a utility counterpart to the theory of simple, approximately revenue-optimal mechanisms.
\end{abstract}

\section{Introduction}

In this paper, we study the design of simple, approximately-optimal mechanisms for consumer utility. This objective is well-motivated by social services, where recipients are often required to pay with their time and effort. For example, patients wait to receive treatment, and applicants for public
benefits complete labor-intensive paperwork~\citep{BravermanCK16,Heflin}. For a service provider who cannot or does not wish to charge money, these
\emph{ordeals} offer a way to ration demand and elicit information about how much people value the service being provided. However, such ordeals are costly to consumers, and provide no tangible benefit to providers. We study 
the design of mechanisms for this setting by modeling nonnegative (non-monetary) payments that are burned---that is, not transferred to the seller. Our objective is hence the buyers' total utility, also called \emph{consumer surplus} or
\emph{residual surplus}: their value for the goods they receive, less the payments they make.

Utility maximization requires a trade-off between maximizing efficiency in the allocation and the payments that are required to truthfully implement such an allocation. For example, for two buyers with similar values, a second-price auction achieves almost no utility, as the payment consumes all of it. Here, a free lottery would be a much better choice. However, if the buyers have a very large gap between their values, the opposite is true; prices are necessary to distinguish the buyer with the high value. In a multi-item setting, the designer also needs to factor in which items each buyer prefers, even if no payment is charged.

The single-dimensional setting is completely resolved, which a Myersonian-like theory by \citet{HartlineR08} giving a closed-form description of the optimal mechanism. \citet{goldnerlundy25} initiate the study of multidimensional Bayesian utility maximization, constructing simple mechanisms for $n$ unit-demand buyers and $m$ items with i.i.d.~values. Their mechanism, which ask buyers what their favorite item is, and then sells each item separately,  gives a tight $\Theta(1 + \log{\frac{n}{m}})$-approximation to the benchmark of social welfare. \citet{ezra2024optimal} extend this approximation to general distributions by parameterizing the amount of item-symmetry present. \citet{goldnerlundy25} leave the open question: what is a better upper-bound on optimal utility than social welfare, and can simple mechanisms guarantee a constant-factor approximation to the optimal utility?

Revenue maximization has been the impetus for decades of foundational work. It is one of the largest and most influential literatures in algorithmic mechanism design. A landmark achievement of this area, reached through a long sequence of breakthroughs, was to show that simple mechanisms can obtain a constant-fraction of the optimal revenue in multidimensional multi-buyer settings---from unit-demand values~\citep{CHMS} to additive values~\citep{Yao15} to a unified framework~\citep{CDW} to XOS values~\citep{CZ17}---despite the complexity of optimal mechanisms. Our result establishes a parallel milestone for consumer utility, providing a foundation for a corresponding theory of simple, approximately-utility-optimal mechanism design.

\paragraph{Our Results.} We develop a general benchmark for optimal expected utility and use it to construct simple mechanisms that obtain a constant fraction of optimal utility for $n$ unit-demand buyers and $m$ heterogeneous items. We assume that all buyer-item values are drawn independently from known distributions. Our contributions are:
\begin{itemize}
\item \emph{A general utility benchmark.}
 We upper-bound optimal ex-ante-constrained utility by a sum of
 optimal single-dimensional utilities and a weighted free lottery
 term (\Cref{thm:single-buyer-benchmark}). This bound applies broadly, whenever
 utility is linear in item 
 parameters, including constrained-additive valuations. For
 unit-demand buyers, it gives a constant-factor approximation to
 optimal utility even when social welfare is a much larger benchmark.
 \item \emph{Symmetric items.} For values drawn i.i.d.~across items and buyers, we give a
 $2e/(e-1)<3.164$ approximation (\Cref{thm:iid-favorite}). The same
 guarantee, up to a multiplicative $1+\varepsilon$, holds when values
 are i.i.d.~across items for each buyer but buyers have different
 distributions (\Cref{thm:iid-stronger}).
 \item \emph{General independent values.} For independent non-identical values,
 we give a
 $2e/(e-1)+(e/(e-1))^2+\varepsilon<5.667+\varepsilon$
 approximation for every $\varepsilon\in(0,1]$
 (\Cref{thm:non-iid-multi}).
\end{itemize}
All of these mechanisms are BIC, ex-post individually rational, and computable in polynomial time. They are incredibly simple: for each buyer, either post a price for every item, or determine a probability for each item with which the item is available for free to the buyer. In both the price and lottery mechanisms, the buyer picks their favorite item given these prices or probabilities, and if multiple buyers choose the same item, contention resolution determines who gets it. The better of these two mechanisms, or a randomized mixture, achieves the above guarantees.

\paragraph{Our Approach: A Single-Buyer Upper Bound.} We decompose utility into the contribution from each buyer. Typically, the revenue maximization literature has first solved the single-buyer problem as a foothold to the multi-buyer problem. However, this poses a problem for utility maximization, as it is the competition between buyers that requires us to use ordeals. For just a single-buyer, we can offer them any item for free.  

Instead, we use the \emph{ex ante relaxation}~\citep{Alaei11,CM,chawla2023buy} to limit how often each item can be allocated. These probabilities come from a multi-buyer mechanism and determine the  probability that item $j$ is sold to bidder $i$ in expectation over \emph{all} of the multi-bidder values profiles $\bf{V}$ that could be drawn. The sum of the optimal utility from each single-buyer ex-ante-constrained problem provides an upper bound on the optimal multi-buyer utility.

With this constraint on allocation, even in a single-buyer problem, the buyer's different possible types now compete for each item's limited allocation. Allocating an item to one type leaves less probability for other types, but incentivizing the buyer to choose an alternative requires costly payments. 
Hence we focus on simple mechanisms that achieve a constant-factor approximation to the optimal ex-ante-constrained utility, and require an upper bound on this quantity.

To form our upper bound, we aim to decompose the multi-item utility into single-item terms. We do so by considering each item's \emph{deletion loss}: how much utility is lost in the multi-item mechanism if we remove the item. We must consider that this term may be large (this item is uniquely the buyer's favorite) or small (there are many other good options for the buyer). 

To account for both possibilities, we consider the utility from two independent draws from the buyer's value distribution, $\Vvec$ and $\Wvec$. Their average has the same expectation as a single draw, so we can write it as
\[
 \E[u(\Vvec)]
 =\E[\min\{u(\Vvec),u(\Wvec)\}]
  +\frac12\E[|u(\Vvec)-u(\Wvec)|].
\]
In \Cref{sec:common-deletion}, we show that this first term is the portion of utility bounded by free lotteries, and this second term is bounded by deletion loss, which we can then bound using item-pricing. \Cref{sec:single,sec:non-iid-single} construct truthful single-buyer multi-item mechanisms that then approximate these terms, and convert them into a multi-buyer mechanism. 

Our approach parallels the core-tail analysis of revenue maximization~\citep{HN13,BILW}, where a cutoff on item value separates items into the outliers (in the tail) and the typical values (in the core), and then simple mechanisms are used to bound each portion. In contrast, rather than divide items based on their realizations, we divide the utility contribution below and above a threshold. This is most similar to a prophet inequality analysis~\citep{Samuel-Cahn,KleinbergW12}. We then show that simple mechanisms bound each portion.

\citet{CDW} gave a unifying framework for generating an upper bound on optimal revenue across multiple valuation classes. There are several impediments to applying that framework here, discussed in \citet{goldnerlundy25}. First, the CDW framework decomposes revenue into single-buyer contributions and generates upper bounds for each separately. As discussed, that does not work in our setting, as the single-buyer problem is trivial. Second, even for the ex-ante-constrained single-buyer problem, the analog of the canonical revenue flows can give bounds no better than welfare. Note also that the copies setting, where each buyer is replaced with single-minded copies, is not necessarily an upper bound on utility~\citep{goldnerlundy25}.

\paragraph{Combining Buyers Using Contention Resolution.} Our multi-buyer mechanism just runs the single-buyer ex-ante-constrained mechanism for each buyer independently. The only problem, then, is when buyers request the same item. We resolve this contention by making every item equally available to each buyer with a constant probability, so the buyers have no incentive to prefer one item over another due to the competition.

When buyer values are drawn symmetrically, resolving contention uniformly-at-random is truthful and guarantees at least a $(1-1/e)$-fraction of utility for each buyer (\Cref{sec:service}). Without item symmetry, for general independent distributions, we must carefully choose how often to offer each item, accounting for different competition across items. A partial Lagrangian upper bound gives multipliers that let us control the demand for each item. With some additional reduction in the probability that a request is served, we regain equal competition for requests across items (\Cref{sec:hetservice}).

\paragraph{Related Work.} The question of single-dimensional Bayesian utility maximization was completely resolved by \citet{HartlineR08}. For multidimensional utility maximization, several works have given logarithmic approximations to the upper bound of social welfare. 
\citet{fotakis2016efficient}
give an $O(\log|\Omega|)$ approximation to welfare, where $\Omega$
is the set of feasible outcomes. More recently,
\citet{ezra2026moneyburning} obtain the tight factor
$H_n=\sum_{k=1}^n1/k$ against welfare for arbitrary nonnegative
valuations over a finite outcome space, and give reductions from
truthful welfare approximations to utility approximations. 

Our mechanism is most similar to the Favorites mechanisms of \citet{goldnerlundy25}: each buyer chooses their favorite item, and a fixed single-item mechanism allocates among buyers who choose the same item. Our mechanism differs in that, rather than being prior-free, it uses known distributions to approximate optimal utility. However, our favorite-item mechanism (Appendix~\ref{sec:combined-bounds}) builds on the Favorites approach, showing that appropriately chosen prices
and contention resolution do in fact give a constant-factor approximation to
optimal utility for item-i.i.d.~values, even when buyers have
different distributions. In contrast to our mechanism, \citet{ezra2024optimal} design the VCG-with-copies mechanism that randomizes the number of copies of each item, runs VCG on that
instance, and rounds the allocation and payment to satisfy the
original supply constraints. Their mechanism is also prior-free. 

Our work follows the approach of designing simple mechanisms
to approximate multidimensional revenue~\citep{CHMS,BILW,CDW} and draws on many tools from this literature.
We use the ex ante relaxation to reduce to single-buyer problems,
as in~\citet{Alaei11,CM,chawla2023buy}, and similarly to~\citet{BMMP24}. Our use of single-dimensional mechanisms to bound the utility of a multi-item mechanism parallels the revenue bounds of~\citet{CHMS,CMS15}.

Our use of deletion to bound optimal utility is conceptually related to \citet{EFFGK}, who in the interdependent values setting, decompose welfare using the value remaining after deleting an agent's signal and the corresponding deletion loss, and bound the resulting contributions using different mechanisms. Related per-signal deletion bounds also appear in \citet{EGZ,EFGMMfocs}. We similarly use per-item deletion to obtain single-dimensional quantities, combining bounds on these deletion losses with a bound on the remaining utility.

There has been a recent surge in the study of utility maximization. This includes utility maximization in market regulation~\citep{EssaidiGW24}, online allocation with predictions~\citep{GMT26}, budget-feasible procurement~\citep{EGKT26}, and dynamic or constrained auctions~\citep{BOT25,BOT26}. Related work in additional domains examine waiting times in healthcare~\citep{BravermanCK16}, verification in resource allocation~\citep{lundy2019allocation,patel2022costly}, utility and retention in mobile games~\citep{lundy2024pay}, and combinatorial pen testing~\citep{GH24}.

\section{Preliminaries}\label{sec:setup}

\subsection{The unit-demand environment and money burning}

There are $n\ge1$ buyers and $m\ge1$ items, with one copy of each
item. Buyer $i$ has value vector $v_i=(v_{i1},\ldots,v_{im})\in\Rplus^m$
and unit-demand value $v_i(S)=\max_{j\in S}v_{ij}$ for a nonempty set
$S\subseteq[m]$, with $v_i(\varnothing)=0$.

For each buyer $i$, values are drawn independently from 
known distributions $V_{ij} \sim F_{ij}$ on $\Rplus$; we use $v_{ij}$ for a realization.
We write $\Vvec_i=(V_{i1},\ldots,V_{im})$ for the random value vector, and
$\Vprofile=(\Vvec_1,\ldots,\Vvec_n)$ for the random instance.
We assume $\E[V_{ij}]<\infty$ throughout, so the expected-utility
benchmarks and the normalizations are finite.

For a reported profile $v$, let $x_{ij}(v)$ be the probability that buyer
$i$ receives item $j$ and let $p_i(v)\ge0$ be their expected payment.
We require a \emph{matching-feasible} allocation: 
each buyer receives at
most one item, and each item is allocated to at most one buyer.
In particular, $\sum_jx_{ij}(v)\le1$ and $\sum_ix_{ij}(v)\le1$. For a single-buyer, only the former constraint remains. 

Buyer $i$'s interim allocation and payment rules average over the other
buyers' values:
\[
 \bar x_i(v_i):=\E_{\Vprofile_{-i}}[x_i(v_i,\Vprofile_{-i})],
 \qquad
 \bar p_i(v_i):=\E_{\Vprofile_{-i}}[p_i(v_i,\Vprofile_{-i})].
\]
A mechanism is Bayesian incentive compatible (BIC) if, for every true
type $v_i$ and alternative report $w_i$,
\[
 u_i(v_i):=v_i\cdot\bar x_i(v_i)-\bar p_i(v_i)
 \ge v_i\cdot\bar x_i(w_i)-\bar p_i(w_i).
\]
Interim individual rationality requires $u_i(v_i)\ge0$.
Our constructed mechanisms satisfy the stronger ex-post requirement:
truthful utility is nonnegative in every realized outcome. For
a single buyer, a mechanism $(x,p)$ is \emph{truthful} (in expectation) if, for every
true value vector $v\in\Rplus^m$ and report $w\in\Rplus^m$,
\[
 u(v):=v\cdot x(v)-p(v)\ge v\cdot x(w)-p(w).
\]
Here $x$ and $p$ average over the mechanism's internal randomness, and 
individual rationality requires $u(v)\ge0$ for every $v$. BIC for multi-buyer mechanisms reduces to this condition
when there is only one buyer.

Our objective is the expected total buyer utility
$\sum_i\E_{\Vvec_i}[u_i(\Vvec_i)]$, also called \emph{residual surplus}.
Let $\OPT$ be its maximum over feasible, BIC, interim individually
rational mechanisms with nonnegative burned payments.
Unless stated otherwise, expectations include the mechanism's internal
randomness. Write $U(\mathcal M)$ for the expected utility of any mechanism $\mathcal M$. For multiple buyers this is total expected utility:
$U(\mathcal M)=\sum_iU_i(\mathcal M)$, where
$U_i(\mathcal M):=\E_{\Vvec_i}[u_i^{\mathcal M}(\Vvec_i)]$.

For a real number $z$, write $z^+:=\max\{z,0\}$ for its
positive part.

\subsection{The ex-ante relaxation}
\label{sec:ex-ante-relaxation}

Let $\SingleOPT(\Vvec_i;\mathbf q_i)$ denote the maximum expected
utility of a truthful, individually rational single-buyer mechanism
with nonnegative payments that allocates at most one item and allocates
item $j$ with ex-ante probability at most $q_{ij}$. We call this the ex ante constraint, or cap. Write $\mathbf q_i=(q_{i1},\ldots,q_{im})$ and $Q_i:=\sum_jq_{ij}$
for the sum of these item caps. When we suppress the buyer index, we
write $\mathbf q=(q_1,\ldots,q_m)$ and $Q:=\sum_jq_j$.
For common item caps $q_j=q$, this gives $Q=mq$.
For independent buyer types, the ex-ante relaxation gives
\begin{equation}\label{eq:ex-ante-relaxation}
 \OPT\le
 \sup_{\substack{\mathbf q_i\in[0,1]^m\ \forall i\\
                  \sum_iq_{ij}\le1\ \forall j}}
 \sum_{i=1}^n\SingleOPT(\Vvec_i;\mathbf q_i).
\end{equation}
The optimal multi-buyer mechanism satisfies the item-supply 
constraints for every realized profile, and therefore also 
in expectation, making its induced single-buyer mechanisms 
feasible for the ex-ante relaxation. 

To distinguish actual allocation probabilities from their
caps, for a single-buyer mechanism $A=(x^A,p^A)$ write
$u^A(v):=v\cdot x^A(v)-p^A(v)$ and
\[
 U(A):=\E_{\Vvec}[u^A(\Vvec)],\qquad
 \eaqj(A):=\E_{\Vvec}[x_j^A(\Vvec)],\qquad
 \eaq(A):=(\eaq_1(A),\ldots,\eaq_m(A)).
\]
For a multi-buyer mechanism
$\mathcal M$, write
$\eaq_{ij}(\mathcal M):=\E_{\Vprofile}[x_{ij}^{\mathcal M}(\Vprofile)]$;
this equals $\eaqj(\mathcal M_i)$ for its interim mechanism
$\mathcal M_i$. 
Actual unit-demand allocation
probabilities always satisfy $\sum_j\eaqj(A)\le1$, even when the sum
of the given upper bounds exceeds 1.

We can view each single-buyer mechanism, including an interim mechanism,
as a menu of allocation lotteries and expected payments, with the option
to receive nothing and pay nothing. If a mechanism is specified only on
the distribution's support, extend it to all $v\in\Rplus^m$ by letting
each type choose a utility-maximizing menu option, taking the menu's
closure if needed. We retain the original choices on the support, so this
extension changes neither expected utility nor allocation probabilities.
In particular, $u(v_{-j},0)$ is well-defined even if zero is outside the
support, and nonnegative payments imply $u(0)=0$.

Note that in our ex-ante-constrained single-buyer mechanisms, we assume buyers break ties between purchasing and buying nothing in favor of buying nothing, as a convenience for satisfying ex-ante caps.

\subsection{The single-dimensional benchmark for a single buyer} \label{sec:sigma}

Let $X$ be a nonnegative single-dimensional value with finite
expectation. A single-dimensional mechanism specifies an allocation
probability $a(z)\in[0,1]$ and an expected burned payment $\tau(z)\ge0$
for each report $z \geq 0$. For $q\in[0,1]$, define
\begin{equation}\label{eq:sigma-definition}
 \sig(X;q):=\max\left\{
  \E[Xa(X)-\tau(X)]:
  (a,\tau)\text{ is BIC and IR},\ \E[a(X)]\le q
 \right\}.
\end{equation}
By a slight abuse of notation, we use $\sig(X;q)$ for both
the optimal mechanism and its expected utility. As is the common convention in Bayesian mechanism design, we assume efficient computation of mechanism $\sig(X;q)$. Note that this mechanism can always be implemented
by drawing a (possibly infinite) posted price independently of the value from at most 2 (price) options; see
Appendix~\ref{sec:compute} for this characterization and computational considerations. 

\section{A single buyer ex-ante utility upper bound}\label{sec:common-deletion}

In this section, we focus on the single-buyer's ex-ante-constrained utility and provide a general upper-bound to then approximate via simple mechanisms. This benchmark holds for buyer distributions with independence across items and when utility can be written $v \cdot x - p$. At the end of the section, we impose a unit-demand setting to show how separate item prices can recover the single-dimensional terms in this bound. 

Since there is only one buyer for this section, we drop buyer subscripts $i$ for ease of reading; subscripts $j$ only refer to items. We then bound the utility of \emph{any} given truthful and individually rational mechanism $(x,p)$ with non-negative payments. As the mechanism is fixed, our analysis is in terms of utility $u(v):=v\cdot x(v)-p(v)$, but note that this is a function of the mechanism $(x,p)$ as well. Let $q_j$ be item $j$'s ex ante item cap, and hence an upper bound on the ex-ante probability of allocating item $j$ in $x$.

\textbf{Decomposing Utility:} We first ask: when is it better to use a free lottery instead of a price? In a lottery, the allocation budget $q$ is divided equally among all types, with no cost. In contrast, setting a price allows those willing to pay it to get a larger-than-$q$ allocation at the expense of those below it while still allocating with probability at most $q$ in expectation over the realization of types. The question is then: when is it necessary to allocate extra to high types, treating them differently?

When the types concentrate, there is little benefit to treating them differently, and a lottery can suffice. However, when there is substantial variation, prices can distinguish and benefit the higher types. This is precisely the intuition behind our utility benchmark. We consider two independent draws $\Vvec$ and $\Wvec$, whose average utility has the same expectation $\E_{\Vvec}[u(\Vvec)]$. For each pair of draws, we decompose their average utility into the \emph{overlapping} utility present in both draws---the smaller of their two utilities---and the \emph{variation} from there to their average. Observe that if utility is constant across draws, such as in a point mass distribution, then all utility is in the overlap term. More generally, taking the minimum removes the excess utility present in only one draw. In essence, the overlap term is then the portion of utility for which a lottery gives a good approximation; the variation term is portion for which we need to use prices to differentiate types.

Note that while different from the classic core-tail analysis in revenue, there are multiple parallels. We too divide our bound into a term that one candidate mechanism can cover and the excess, then bounded by the other candidate mechanism. Concentration or lack thereof helps explain the difference between the terms.

Formally, with independent draws $\Vvec$ and $\Wvec$ from the buyer's
value distribution, we get the following decomposition

\begin{equation}\label{eq:utility-base}
\begin{aligned}
 \E_{\Vvec}[u(\Vvec)]
 &=\frac{\E_{\Vvec}[u(\Vvec)]+\E_{\Wvec}[u(\Wvec)]}{2}\\
 &=\underbrace{\E_{\Vvec,\Wvec}
       \bigl[\min\{u(\Vvec),u(\Wvec)\}\bigr]}_{\text{overlap term}}
   +\underbrace{\frac12\E_{\Vvec,\Wvec}
       \bigl[|u(\Vvec)-u(\Wvec)|\bigr]}_{\text{variation term}}.
\end{aligned}
\end{equation}

\textbf{Overlap term:} One reason expected utility may be high and difficult to upper-bound is because of correlation between the values and allocation rules: when a buyer's values are higher, their allocation may also be higher. For the first term in our decomposition, the minimum utility of two independent value draws $\Vvec$ and $\Wvec$, we show that this is upper-bounded by replacing the correlation with independence, using the value and allocation from independent draws.

More precisely, we show that this component of utility is well-approximated by a lottery. Using the AM-GM inequality gives that this is upper-bounded in expectation by allocating to a bidder with value $\Vvec$ with the allocation probability of a fresh draw $\Wvec$, eliminating the correlation between value and allocation. 
Taking expectations, independence lets us upper-bound each component separately by the expected highest item value multiplied by the expected allocation probability, which is at most the sum of the ex-ante allocation caps across items, $Q$.

For ex ante caps stemming from a unit-demand mechanism, $Q \leq 1$, and this is precisely the expected utility from giving the bidder their favorite item for free with probability $Q$. Hence the expected minimum utility of two draws is upper-bounded by the ex-ante constrained free lottery. In the non-i.i.d.~setting, we bound the term using item-specific free lotteries, where each item's probability of allocation is chosen in order to balance utility against demand.

\textbf{Variation term:} The variation term measures the expected difference in utility between independent draws, and we show that this component of utility is well-approximated by a item prices. By symmetry, half of the expected absolute difference between $u(\Vvec)$ and $u(\Wvec)$ equals the expected positive part of $u(\Vvec)-u(\Wvec)$. To bound this difference, we consider the more extreme effect of removing each item's 
value entirely, when positive. That is, we replace the coordinates of $\Wvec$ with those of $\Vvec$, one at a time, and observe that the total increase in utility is at most the sum of the positive increases at each step. Since $W_j \geq 0$, increasing coordinate $j$ from 0 to $V_j$ can only give a larger gain than replacing $W_j$ with $V_j$. 
This gain is exactly the loss from zeroing out coordinate $j$ in the resulting vector. Independence ensures that each intermediate vector has the same distribution as $\Vvec$, so its expected \emph{deletion loss} is $D_j$. Summing over coordinates bounds the variation term by $D$. 

Formally we define the total expected deletion loss $D$ of a 
mechanism as
\begin{equation}\label{eq:deletion}
 D:=\sum_{j=1}^m D_j,
 \qquad D_j:=\E_{\Vvec}\bigl[u(\Vvec)-u(\Vvec_{-j},0)\bigr].
\end{equation} 
We then bound each of these terms to establish our utility benchmark in \Cref{thm:single-buyer-benchmark}. 
In \Cref{lem:deletion-coordinate}, we show how to attain expected utility equal to $D_j$ via an implementable single-item mechanism, hence upper-bounding $D$ by the sum of utility from each optimal ex-ante-constrained single-item mechanism. Finally, in \Cref{lem:combine-independent-prices}, we show that running the per-item mechanisms simultaneously to a unit-demand buyer yields a constant-fraction of the utility from each of the single-item mechanisms separately, where there is no competition among items.

\begin{lemma}\label{lem:weighted-residual}
Consider any truthful, individually rational single-buyer
mechanism with nonnegative payments, expected utility $U$, total expected
deletion loss $D$, and ex-ante probability of allocating item $j$ at most
$q_j$. For any $\ell_1,\ldots,\ell_m \in(0,1]$,
\begin{equation}\label{eq:weighted-residual}
 U\le 
 \E_{\Vvec}\!\left[\max_j \ell_jV_j\right]\sum_j\frac{q_j}{\ell_j} + D.
\end{equation}
\end{lemma}
 
The parameters $\ell_j$ describe free item lotteries: choosing item $j$ gives it 
to the buyer with probability $\ell_j$, so their expected utility from the best lottery is 
$\max_j \ell_jV_j$.

\begin{proof}

We bound the overlap and variation terms in \Cref{eq:utility-base} separately.

First, observe that utility is of course upper-bounded by welfare as payments are nonnegative. 
Then, for each item $j$, we can rewrite its contribution to the \emph{welfare}, $v_j \cdot x_j(v)$ as the \emph{utility} of a free lottery that selects $j$ with probability $\ell_j$, namely $\ell_j v_j$, scaled by the weight  $x_j(v)/\ell_j$. Summing over items and replacing each lottery's utility by the utility of the buyer's favorite lottery gives an upper bound: the favorite lottery's utility multiplied by the total scaling weight $\sum_jx_j(v)/\ell_j$.

Formally, individual rationality implies
\begin{equation}\label{eq:pointwise}
 0\le u(v)\le \sum_j v_j\cdot x_j(v) = \sum_j \left(\ell_j v_j\right)\cdot \left(\frac{x_j(v)}{\ell_j}\right)
 \le\left(\max_j \ell_jv_j\right)
       \left(\sum_j\frac{x_j(v)}{\ell_j}\right).
\end{equation}

As we know from \citep{goldnerlundy25}, a welfare bound for the full expected utility can be too large to yield a constant approximation. However, by applying \Cref{eq:pointwise} only to the overlap term, which is smaller than $\E[u(\Vvec)]$, we get a small enough bound:

\begin{align}
 \E&_{\Vvec,\Wvec}\bigl[\min\{u(\Vvec),u(\Wvec)\}\bigr] \\
 &\le \E_{\Vvec,\Wvec}\bigl[\sqrt{u(\Vvec) \cdot u(\Wvec)}\bigr] \notag & u(\Vvec),u(\Wvec) \geq0 \\
 &\le \E_{\Vvec,\Wvec}\!\left[
 \sqrt{\left(\max_j \ell_jV_j\right)
       \left(\sum_j\frac{x_j(\Vvec)}{\ell_j}\right)
       \cdot
       \left(\max_j \ell_jW_j\right)
       \left(\sum_j\frac{x_j(\Wvec)}{\ell_j}\right)}\right]\notag & \eqref{eq:pointwise}\\
 &\le \frac12\E_{\Vvec,\Wvec}\!\left[
       \left(\max_j \ell_jV_j\right)\sum_j\frac{x_j(\Wvec)}{\ell_j}
       +\left(\max_j \ell_jW_j\right)\sum_j\frac{x_j(\Vvec)}{\ell_j}
       \right]\notag & \text{AM-GM}\\
&=\frac12\left(
   \E_{\Vvec}\!\left[\max_j \ell_jV_j\right]
   \E_{\Wvec}\!\left[\sum_j\frac{x_j(\Wvec)}{\ell_j}\right]
   +\E_{\Wvec}\!\left[\max_j \ell_jW_j\right]
   \E_{\Vvec}\!\left[\sum_j\frac{x_j(\Vvec)}{\ell_j}\right]
   \right)\notag & \text{independence}\\
 &=\E_{\Vvec}\!\left[\max_j \ell_jV_j\right]
   \E_{\Vvec}\!\left[\sum_j\frac{x_j(\Vvec)}{\ell_j}\right] \notag\\
 &\le\E_{\Vvec}\!\left[\max_j \ell_jV_j\right]\sum_j\frac{q_j}{\ell_j}, & \E_{\Vvec}[x_j(\Vvec)]\le q_j
 \label{eq:free-bound}
\end{align}
where the two equalities use that the two vector draws are independent and identically distributed.

To bound the variation term, we replace the coordinates of $\Wvec$ with those of $\Vvec$, one at a time. Let $H^{(j)}:=(V_1,\ldots,V_j,W_{j+1},\ldots,W_m)$ be the vector after the first $j$ replacements, so $H^{(0)}=\Wvec$ and $H^{(m)}=\Vvec$. The changes in utility telescope, and discarding negative changes gives an upper bound:
\begin{align}
 \frac{1}{2}\E_{\Vvec,\Wvec}\bigl[\big|u(\Vvec)-u(\Wvec)\big|\bigr]
 &=\E_{\Vvec,\Wvec}\bigl[(u(\Vvec)-u(\Wvec))^+\bigr]   \notag & \text{symmetry}\\
 & =\E_{\Vvec,\Wvec}\left[
  \left(u(H^{(m)})-u(H^{(0)})\right)^+
\right]\notag & \text{def. of $H^{(j)}$}\\
 &=\E_{\Vvec,\Wvec}\left[
   \left(\sum_{j=1}^m
   \bigl(u(H^{(j)})-u(H^{(j-1)})\bigr)\right)^+
 \right]\notag & \text{telescoping}\\
 &\le\E_{\Vvec,\Wvec}\left[
   \sum_{j=1}^m
   \bigl(u(H^{(j)})-u(H^{(j-1)})\bigr)^+
 \right]\notag & (\sum_j a_j)^+\le\sum_j a_j^+\\
 &\le\sum_{j=1}^m
   \E_{\Vvec}\bigl[u(\Vvec)-u(\Vvec_{-j},0)\bigr]\notag & \text{monotonicity of $u$}\\
 &=D, & \text{def. of $D$}
 \label{eq:hybrid}
\end{align}
where the final inequality follows because $u$ is coordinatewise nondecreasing and $H^{(j)}$ has the 
same distribution as $\Vvec$. 

Together, these bounds give \Cref{eq:weighted-residual}, thus proving the lemma.
\end{proof}

Having upper-bounded the variation term by the deletion loss, we now show that, for each item $j$, the deletion loss $D_j$ is the expected utility attained by a single-item mechanism with ex-ante allocation probability at most $q_j$. Thus, $D_j$ is upper-bounded by the \emph{optimal} single-item ex-ante-constrained utility $\sig(V_j;q_j)$.

\begin{lemma}\label{lem:deletion-coordinate}
For any truthful, individually rational single-buyer mechanism
with nonnegative payments and utility $u$, let $D_j$ be its
expected deletion loss for item $j$, and let $q_j$ be an upper bound on its
ex-ante probability of allocating that item. Then, for every item $j$,
\[
 D_j\le\sig(V_j;q_j).
\]
\end{lemma}
\begin{proof}
For a fixed item $j$, we construct a mechanism that takes 1-dimensional input $V_j$ by resampling fresh draws to use as values for the other items.

That is, draw fresh samples $Z_k\sim F_k$ independently for every item $k\ne j$; we write $\Zvec_{-j}:=(Z_k)_{k\ne j}$. The agent then reports their bid $b_j$ for item $j$, and averaging over these draws 
$\Zvec_{-j}$, we define the function
\[ d_j(b_j):=\mathbb E_{\mathbf Z_{-j}}
 \bigl[u(\mathbf Z_{-j},b_j)-u(\mathbf Z_{-j},0)\bigr]. \]
For a realized value $v_j$, $d_j(v_j)$ is the expected utility lost by deleting item $j$, averaging over the other item values. A single-item mechanism whose truthful utility is $d_j(v_j)$ therefore has expected utility $D_j$. Indeed, taking expectations over $V_j$ gives
\[
 \E_{V_j}[d_j(V_j)]
 =\E_{V_j,\Zvec_{-j}}[u(\Zvec_{-j},V_j)-u(\Zvec_{-j},0)]
 =\E_{\Vvec}[u(\Vvec)-u(\Vvec_{-j},0)]
 =D_j.
\]

We now define the single-item mechanism $(a, \tau)$. The allocation rule $a$ averages the original mechanism's allocation probability $x_j$ over the sampled values of the other items: 
$a(b_j):=\mathbb E_{\mathbf Z_{-j}}
 \bigl[x_j(\mathbf Z_{-j},b_j)\bigr]$.
Then in expectation, it continues to satisfy the ex-ante allocation cap $q_j$, as
\[
 \E_{V_j}[a(V_j)]
 =\E_{V_j,\Zvec_{-j}}[x_j(\Zvec_{-j},V_j)]
 =\E_{\Vvec}[x_j(\Vvec)]\le q_j.
\]
We choose the payment so that truthful reporting gives expected utility $d_j(v_j)$, the average deletion loss for $v_j$: define $\tau(b_j):=b_ja(b_j)-d_j(b_j)$.

Now, for the bidder's true single-dimensional value $v_j \geq 0$ and for any alternate bid $b_j\ge0$, truthfulness of the original single-buyer mechanism that 
implements $u$ gives 
\[
\begin{aligned}
 u(\Zvec_{-j},v_j)
 & = v_j\,x_j(\Zvec_{-j},v_j)
      +\sum_{k\ne j}Z_kx_k(\Zvec_{-j},v_j)-p(\Zvec_{-j},v_j)\\
 &\ge v_j\,x_j(\Zvec_{-j},b_j)
      +\sum_{k\ne j}Z_kx_k(\Zvec_{-j},b_j)-p(\Zvec_{-j},b_j) & \text{IC}\\
 &=\left(b_j\,x_j(\Zvec_{-j},b_j)
      +\sum_{k\ne j}Z_kx_k(\Zvec_{-j},b_j)-p(\Zvec_{-j},b_j)\right)
      +(v_j-b_j)x_j(\Zvec_{-j},b_j)\\
 &=u(\Zvec_{-j},b_j)+(v_j-b_j)x_j(\Zvec_{-j},b_j),
\end{aligned}
\]
where the final two equalities add and subtract the reported value's contribution to
item $j$ and use the definition of $u$.

Subtracting $u(\Zvec_{-j},0)$ from both sides and taking expectations over $\Zvec_{-j}$ gives 
\begin{equation}
 d_j(v_j)\ge d_j(b_j)+(v_j-b_j)a(b_j).
 \label{eq:deletion-support}
\end{equation} 
Since $d_j(0)=0$ and $a(b_j)\ge0$, inequality~\eqref{eq:deletion-support} establishes the required incentive and payment properties:
\begin{description}
  \item[\normalfont\emph{Truthfulness}:] For true value $v_j$ and alternative
  report $b_j$,
  \begin{align*}
    v_j a(v_j)-\tau(v_j)
    &=d_j(v_j)\\
    &\ge d_j(b_j)+(v_j-b_j)a(b_j)\\
    &=v_j a(b_j)-\tau(b_j),
  \end{align*}  \item[\normalfont\emph{Individual rationality}:] Setting $b_j=0$ gives
  \[
    d_j(v_j)\ge d_j(0)+v_j a(0)=v_j a(0)\ge0,
  \]  \item[\normalfont\emph{Nonnegative burned payments}:] Setting $v_j=0$ and using $d_j(0)=0$ gives
  \[
    \tau(b_j)=b_j a(b_j)-d_j(b_j)\ge0.
  \]
\end{description}

We have thus constructed a single-dimensional mechanism with expected utility $D_j$
and ex-ante probability of allocation at most $q_j$. 
Since $\sig(V_j;q_j)$ is the optimal expected utility among such mechanisms, we conclude that
\[
 D_j\le\sig(V_j;q_j).
\]
\end{proof}

Together, the two lemmas immediately give an upper bound on the buyer's
ex-ante-constrained utility stated in \Cref{thm:single-buyer-benchmark}. Neither bound requires the mechanism
to allocate at most one item: both apply whenever utility is linear
in independent type parameters and payments are nonnegative. For example, they hold for any valuations that are additive up to some service constraint.

\begin{theorem}[Single-buyer ex-ante utility bound]
\label{thm:single-buyer-benchmark}
Consider a buyer with independent nonnegative type parameters
$\Vvec=(V_1,\ldots,V_m)$ and a feasible outcome set
$\mathcal F\subseteq[0,1]^m$. A mechanism specifies an expected
outcome $x(b)\in\mathcal F$ and an expected payment $p(b)$ for each
report $b\in\Rplus^m$. A buyer of type $v$ who reports $b$ obtains
expected utility $v\cdot x(b)-p(b)$.

For $\mathbf q\in[0,1]^m$, let
$\SingleOPT_{\mathcal F}(\Vvec;\mathbf q)$ be the supremum expected
utility over truthful, individually rational mechanisms with
nonnegative payments satisfying $\E_{\Vvec}[x_j(\Vvec)]\le q_j$
for every $j$. Then, for any $\ell_1,\ldots,\ell_m\in(0,1]$,
\begin{equation}\label{eq:general-utility-benchmark}
 \SingleOPT_{\mathcal F}(\Vvec;\mathbf q)
 \le
 \sum_{j=1}^m\sig(V_j;q_j)
 +
 \E_{\Vvec}\!\left[\max_j\ell_jV_j\right]
 \sum_{j=1}^m\frac{q_j}{\ell_j}.
\end{equation}
\end{theorem}

The proofs of all three lemmas in this section use independence across items. Lemma~\ref{lem:weighted-residual} also uses nonnegative payments; allowing negative payments substantially changes the utility maximization problem. Both \Cref{lem:weighted-residual,lem:deletion-coordinate} use that utility is linear in the type parameters, but neither requires unit-demand allocations. 

To approximate this bound with the utility of an implementable multidimensional mechanism, we compare mechanisms $\jpp$ for the single-item environment (where the bidder is only interested in item $j$) with a mechanism $\mpp$ for the multi-item environment (where, for unit-demand values, the buyer is only interested in purchasing at most one item, which may or may not be item $j$). The following lemma shows that by combining the single-item posted price mechanisms into a multi-item posted price mechanism with unit-demand constraints, not much utility is lost, and the multi-item mechanism continues to satisfy the given ex ante constraints. To recover the single-item utilities in the benchmark using a feasible mechanism, we recall the posted-price representation of $\sig(V_j;q_j)$ from \Cref{sec:sigma}: it is implemented by posting a randomly drawn price over support at most 2.

\begin{lemma}\label{lem:combine-independent-prices}
Let $V_1, \dots, V_m$ be independent item values,
and $P_j$ be an independent (possibly randomized) nonnegative price $j$ for which $\Pr_{V_j,P_j}(V_j>P_j)\le q_j$, 
which for the one-item setting yields expected utility $U(\jpp)$. If $\sum_jq_j\le1$,
then the mechanism $\mpp$ that offers these prices together and lets 
the buyer choose their utility-maximizing option allocates item $j$ 
with probability at most $q_j$ and has expected utility
\[
 U(\mpp)\ge(1-1/e)\sum_jU(\jpp).
\]
\end{lemma} 

To parse this statement, observe that the assumption $\Pr_{V_j,P_j}(V_j>P_j)\le q_j$ corresponds to posted price for the single-item setting satisfying the ex ante constraint $q_j$. The assumption that $\sum_j q_j \leq 1$ holds when the ex ante constraints $q_j$ are given by the multi-item setting, where we have a feasibility constraint that the mechanism allocates at most one item, which is without loss for unit-demand valuations.

\begin{proof}
Let $\mathbf P=(P_1,\ldots,P_m)$ be independent (possibly randomized) nonnegative posted prices for each item. Consider first the single-item setting, where only item $j$ is for sale. Define the surplus from posted price $P_j$ as $S_j:=(V_j-P_j)^+$. Then the expected utility of posted price $P_j$ is $U(\jpp) :=\E_{V_j,P_j}[S_j]$. 

Now consider the multi-item posted price mechanism $\mpp$ with a unit-demand feasibility constraint. That is, post prices $\mathbf{P}$ and let the buyer choose one item. In this setting, let $\eaqj(\mpp)$ denote the ex ante probability that mechanism $\mpp$ allocates item $j$. Observe that it is bounded by the probability of sale in the single-item posted price setting, and thus also the single-item cap $q_j$: 
\[
 \eaqj(\mpp)
 \leq \Prb_{\Vvec,\mathbf S}\!\left(S_j>0\ \text{and}\ S_j=\max_k S_k\right)
 \leq\Prb_{V_j,P_j}(S_j>0)
 \leq q_j.
\]

Finally, we show that the utility of the multi-item mechanism is a good approximation to the sum of that from the single-item mechanisms. For $t\ge0$, define the expected number of items whose surplus exceeds $t$ as 
$L(t):=\sum_j\Prb_{V_j,P_j} (S_j>t)$, which by the single-item ex ante constraints, is at most $\sum_jq_j\le1$. We bound the utility of the multi-item posted prices as follows.
\begin{align*}
 U(\mpp)
 =\E_{\Vvec,\mathbf S}\!\left[\max_jS_j\right] 
 &=\int_0^\infty
       \left(1-\prod_j\bigl(1-\Prb_{V_j,P_j}(S_j>t)\bigr)\right)\,dt & \text{independence}\\
 &\ge\int_0^\infty(1-e^{-L(t)})\,dt & 1-x\le e^{-x}\\
 &\ge(1-1/e)\int_0^\infty L(t)\,dt & 1-e^{-s}\ge(1-1/e)s, \, s \in [0,1]\\
 & = (1-1/e) \sum_j \int_0^\infty \Prb_{V_j,P_j}(S_j>t)\,dt \\
 &=(1-1/e)\sum_jU(\jpp).
\end{align*}
\end{proof}

\section{A constant-factor approximation for i.i.d.~buyers}\label{sec:single}

In this section, we provide a constant-factor approximation for the i.i.d.~setting: when every buyer's value for every item is drawn independently and identically from the same distribution, that is, $V_{ij}\sim F$ independently for all $i,j$, and let $V\sim F$.
First, we use the bound from Section~\ref{sec:common-deletion} to construct an approximately optimal single-buyer mechanism (\Cref{sec:iid-single}). 
Posted item prices recover the single-item utility terms, and free
allocation recovers the overlap term. This construction serves as a warm-up for \Cref{sec:non-iid-single}, which uses a similar mechanism, and also follows the steps of a two-part bound. 
 We then combine buyers
through contention resolution (\Cref{sec:service}), first stating the
general condition that preserves incentives and then verifying it for
uniform contention resolution in the i.i.d.~setting. Finally, we simplify
the single-buyer mechanism, improve its guarantee, and extend the result
to buyers whose values are i.i.d.~across items but may have different
distributions across buyers (\Cref{sec:combined-bounds}).

We use $Q$ for a buyer's total ex-ante allocation cap.  For i.i.d.~item values,
symmetrizing a mechanism over item labels preserves expected utility and
makes its allocation probabilities equal across items, hence we set a common
item cap at $q:=Q/m$, with $\mathbf q=(q,\ldots,q)$. Thus the optimum
under a total cap $Q$ equals $\SingleOPT(\Vvec;\mathbf q)$, using the
item-cap notation of \Cref{sec:setup,sec:common-deletion}.
We leave $Q\in[0,1]$ arbitrary in the single-buyer argument and set
$Q=\min\{1,m/n\}$ when combining i.i.d.~buyers.

\subsection{The single-buyer mechanism}\label{sec:iid-single}

We build directly on \Cref{thm:single-buyer-benchmark}. 
For any total ex-ante cap $Q\in[0,1]$, with common item cap $q:=Q/m$,
define the single-buyer mechanism $\singleiidmech$ as follows:

\begin{enumerate}
 \item With probability $e/(2e-1)$, run $\mpp$: Draw each item's price independently
 according to $\sig(V;q)$. Let the buyer choose the item with the largest 
 positive utility (if any) and charge its price.
 \item With probability $(e-1)/(2e-1)$, run $\mfree$: Let the buyer choose their favorite item
 and allocate it to them with probability $Q$ free of charge.
\end{enumerate}

To implement $\mpp$, we use the posted-price representation of $\sig(V;q)$
from \Cref{sec:sigma}. This gives a distribution over at most two prices, and which we note is efficiently computable (see Appendix~\ref{sec:compute}).

Note that $\mfree$ is not simply $\mpp$ with the price set to 0, as it allocates the favorite item with the combined probability $Q$ rather than only $q$, which still satisfies the per-item cap $q$ in expectation over the buyer's favorite item.

We now show that the utility of our mechanism $\singleiidmech$ is a constant-factor approximation to the optimal single-buyer ex-ante-constrained utility: its components bound the terms discussed in \Cref{thm:single-buyer-benchmark}.

\begin{lemma}\label{lem:single}
For i.i.d.~item values and every $Q\in[0,1]$, with $q=Q/m$ and
$\mathbf q=(q,\ldots,q)$, mechanism $\singleiidmech$ is truthful,
ex-post individually rational, has nonnegative payments, allocates each
item with ex-ante probability at most $q$, and satisfies
\[
 \utilsingleiidmech\ge\frac{e-1}{2e-1}\SingleOPT(\Vvec;\mathbf q).
\]
\end{lemma}

\begin{proof}
Apply \Cref{thm:single-buyer-benchmark} with $\ell_j=1$ and $q_j=q$ for
every item. Since $\sum_jq_j=Q$, we obtain
\begin{equation}\label{eq:iidsingleDF}
 \SingleOPT(\Vvec;\mathbf q)
 \le m\sig(V;q)+Q\E[\max_jV_j]
 =m\sig(V;q)+\UF.
\end{equation}
The equality holds because $\mfree$ gives the buyer their favorite item
with probability $Q$, independently of their values.
The prices in $\mpp$ attain $\sig(V;q)$ for each item separately, and as $mq=Q\le1$, then \Cref{lem:combine-independent-prices} gives
\begin{equation}\label{eq:deletion-price}
 m\sig(V;q)\le\frac{e}{e-1}\UP.
\end{equation}
Combining \Cref{eq:iidsingleDF,eq:deletion-price} yields
\[
 \SingleOPT(\Vvec;\mathbf q)\le\frac{e}{e-1}\UP+\UF.
\]
Using the probabilities of $\singleiidmech$ for running each sub-mechanism, we get
\[
 \utilsingleiidmech
 =\frac{e}{2e-1}\UP+\frac{e-1}{2e-1}\UF
 \ge\frac{e-1}{2e-1}\SingleOPT(\Vvec;\mathbf q).
\]
Both sub-mechanisms satisfy the item caps and let the buyer choose her
utility-maximizing option. Their combination is therefore
truthful, ex-post individually rational, and has nonnegative payments.
\end{proof}

We then conclude that the better (utility-maximizing) of these two mechanisms gives a constant-factor approximation. 

\begin{corollary}
Under the assumptions of \Cref{lem:single},
\[
 \max\{\UP,\UF\}
 \ge \frac{e-1}{2e-1}\SingleOPT(\Vvec;\mathbf q).
\]
\end{corollary}

All of our guarantees then apply to both the randomized $\singleiidmech$ and a replacement that uses only the utility-maximizing sub-mechanism.

\paragraph{Simpler, Stronger, and More General Guarantees.} The above mechanism and analysis is provided for its simplicity and its parallels to \Cref{sec:non-iid-single}. However, when buyers have identical, symmetric distributions for their value for each item, that is, in the \emph{item-i.i.d.}~setting, we can use an even simpler mechanism: compute the 1-item (randomized) posted price for the buyer's distribution of their value for their favorite item. Post this price for every item. For fully i.i.d.~buyers, uniform
contention resolution gives expected utility at least
$(1-1/e)\OPT/2$. For buyers with different distributions, use the
contention rule and computed caps in \Cref{sec:buyer-specific} to
obtain the guarantee of \Cref{thm:iid-stronger}.

\subsection{The multi-buyer mechanism}\label{sec:service}

Simultaneously running $\singleiidmech$ for each bidder may result in a single item being allocated to multiple bidders. In this section, we resolve the contention caused by multiple requests for an item, and regain matching-feasibility in the multi-buyer mechanism.

We first state the contention resolution argument used here and in
\Cref{sec:non-iid-single}. Run each single-buyer mechanism independently,
treating its allocation as tentative and collecting payment only if the
buyer ultimately receives the item. If a buyer is tentatively allocated item $j$, we call this a ``request''     to item $j$. For a free lottery with allocation
probability $Q$, toss the allocation coin after the buyer chooses an item
and before resolving contention; convert this to a request only if
the coin succeeds. Buyers make their choices without observing the
tentative allocations or random outcomes for other buyers.

We follow a common approach of contention resolution in showing that each request is served with constant probability $c_i$.

\begin{observation}\label{lem:utility-preserving-contention}
Suppose each buyer $i$'s request is served with probability $c_i$,
regardless of their report or tentative outcome, and charges their only when served.
This multiplies the utility of every report by $c_i$, preserving incentive
compatibility and giving total expected utility $\sum_i c_i U_i$, where
$U_i$ is their utility before contention.
Serving at most one request per item ensures feasibility; individual
rationality and nonnegative payments are inherited from the single-buyer mechanisms.
\end{observation}

\Cref{lem:utility-preserving-contention} does not require
identical distributions, only independence.

To identify $c_i$'s for the i.i.d.~case, we leverage that by item symmetry, each item has the same distribution of competing requests. Resolving conflicts uniformly-at-random then gives the same $c_i$ for each item. 

We now instantiate the argument for i.i.d.~buyers. Fix
\[
 Q:=\min\{1,m/n\},\qquad q:=Q/m=\min\{1/m,1/n\},
 \qquad \mathbf q=(q,\ldots,q).
\]
The multi-buyer mechanism $\iidmech$ is as follows:
\begin{enumerate}
 \item Run $\singleiidmech$ independently for each buyer, treating its
 allocation as a request. In the free sub-mechanism, a request is made only
 if the probability-$Q$ allocation coin toss succeeds.
 \item For each item with at least one request, serve one requester
 uniformly at random and charge them the corresponding price, if any.
\end{enumerate}
The allocation caps ensure expected demand of at most one per item. We next show that
uniform contention resolution serves each buyer's request with probability at
least $1-1/e$. 
By item-symmetry, a buyer faces the same competition no matter which item they select. The given (optimal) ex-ante allocation caps ensure that the expected demand for each item is at most 1. Together, this implies that uniform contention resolution gives each request a constant probability of succeeding, as needed by \Cref{lem:utility-preserving-contention}.

\begin{restatable}{theorem}{iidsimple}\label{thm:iid-simple}
For i.i.d.~values, mechanism $\iidmech$ is BIC, ex-post
individually rational, has nonnegative payments, respects matching-feasibility,
and achieves expected utility
\[
 U(\iidmech) \geq \frac{e-1}{2e-1}\left(1-\frac{1}{e}\right)\OPT.
\]
\end{restatable}

\begin{proof}
Symmetry and \Cref{lem:single} give that before contention resolution, utility is at least $\frac{e-1}{2e-1}\OPT$.
Since $q\le1/n$, the Favorites argument of \citet[Theorem~1]{goldnerlundy25} gives that each item is available to be served to each buyer with probability at least $1-1/e$.
The result follows from \Cref{lem:utility-preserving-contention}.
\end{proof}

\section{Independent item values across multiple buyers}
\label{sec:non-iid-single}

We now allow every $V_{ij}\sim F_{ij}$ to have a different distribution,
retaining independence across items and buyers. The single-buyer bounds
of \Cref{sec:common-deletion} still apply and will drive the approximation 
mechanism. We emphasize that, even with access to the optimal mechanism's 
ex-ante probability caps\footnote{The ex-ante relaxation is typically formulated as a convex program, assuming computational access to the single-buyer benchmarks~\citep{Alaei11}. This access is nontrivial: even with independent item values, the linear program can be exponential in the number of items.}, it is not immediately clear how one can efficiently construct single-buyer mechanisms that recovers the benchmark's free-allocation 
term and abides by the desired item caps. 

We reuse the two sub-mechanisms from  \Cref{sec:iid-single}: posted item prices and free item lotteries. Without item symmetry, certain items may have higher demand and thus contention between buyers. We therefore choose the prices and lottery probabilities using a demand weights
for each item, increasing it when that item receives too many requests. Averaging these choices gives the independent single-buyer
mechanisms needed for the contention argument of \Cref{sec:hetservice}. 

\subsection{Demand weights and balancing item requests} \label{sec:het-single}

For the multi-buyer optimization problem of $\OPT$, consider the partial Lagrangian primal that moves the ex-ante allocation constraints into the objective, multiplying each by dual variables $\lambda_j$. In essence, for any mechanism $\mathcal{M}$, this objective penalizes (or rewards) over-allocation (or under-allocation) of item $j$ with weight $\lambda_j$. For a fixed choice of dual variables, we can lower-bound the item-feasibility term (corresponding to the constraint that there's at most 1 copy of each item) with the ex-ante allocation probabilities $\eaq(C)$ of any feasible multi-buyer unit-demand mechanism $C$, since $\sum_i \eaq_{ij}(C) \leq 1$ for all $j$. This only moves the baseline by a constant. We then decompose it into single-buyer terms. Then for a fixed single-buyer $i$, this results in our single-buyer proxy-Lagrangian $\G^i(\mathcal{M}_i,C_i)$. 
That is, 
\begin{align}
 \mathcal{L}_{\lambda}(\mathcal{M}) &= \sum_i U_i(\mathcal{M}) +\sum_j\lambda_j\left(1-\sum_i\eaq_{ij}(\mathcal{M})\right)
       \notag \\
 &\ge \sum_i U_i(\mathcal{M}) +\sum_j\lambda_j\left(\sum_i\eaq_{ij}(C)-\sum_i \eaq_{ij}(\mathcal{M})\right)
       \notag & \text{$\sum_i\eaq_{ij}(C) \le1$, $\lambda_j\ge0$}\\
 &:=\sum_i \G^i(\mathcal{M}_i,C_i).                   
       \notag & \text{definition of $\G$}
\end{align}
where a buyer subscript or superscript
$i$ on a mechanism or objective selects its single-buyer component. Note that the value distributions are fixed throughout this section, so we suppress the dependence on $\Vvec$. 
We use $\G^i$ as a comparison term for the utility from $\mathcal M_i$ and $C_i$. It is the utility from $\mathcal M_i$, adjusted for how its demand differs from $C_i$'s demand, and we use it to move between the contribution of multi-item and single-item mechanisms. 

We focus on the single-buyer problem, fixing one buyer $i$ and suppress their index. Then for any single-buyer unit-demand mechanisms $\mathcal{M}$ and $C$, we define this single-buyer Lagrangian lower bound as: 
\begin{equation*}
 \G(\mathcal{M}, C):=U(\mathcal{M})-\sum_j\lambda_j \eaqj(\mathcal{M}) + \sum_j\lambda_j \eaqj(C).
 \end{equation*}
 Note that $C$ only requires unit-demand feasibility, that $\sum_j \eaqj(C) \leq 1$. The cross-buyer item-feasibility is only needed for the Lagrangian lower bound above, and not for future comparison steps. We term this objective mechanism $\mathcal{M}$'s \emph{demand-adjusted} utility for $C$, and we refer to the dual variables $\lambda_j$ as \emph{demand weights}. 
 Here, $C$ is the mechanism whose utility we want to approximate using $\M$. Its allocation probabilities provide a comparison baseline for $\M$, but do not serve as caps for $\M$.  To give the multi-buyer guarantee, we will choose comparison mechanisms whose total utility is at least $\OPT$, and whose aggregate demands satisfy item feasibility (where supply is at most 1). We use the comparison $C$ only in the analysis: the construction of $\M$ depends on the value distributions and $\lambda$, not on $C$ or its caps.
 
Parallel to \Cref{sec:iid-single}, we design a single-buyer mechanism with two sub-mechanisms to bound both parts of \Cref{thm:single-buyer-benchmark}. 
Our mechanism $\singlehetmech$ mixes posted item prices $\mpp$ with free item lotteries $\mfree$, where prices and lottery probabilities are chosen using the demand weights. See the formal definition at the end of the section. 

At any fixed demand weights $\mathbf{\lambda}$, the following lemma gives one mechanism whose demand-adjusted utility is comparable to that of any other mechanism.

\begin{lemma}\label{lem:non-iid-local}
For independent item values and any nonnegative demand-weight vector
$\lambda=(\lambda_1,\ldots,\lambda_m)$, there exists a BIC, ex-post
individually rational single-buyer unit-demand mechanism $A$ with
nonnegative payments such that, for every BIC, individually rational
single-buyer unit-demand mechanism $C$ with nonnegative payments and
ex-ante allocation probabilities $\eaqj(C)$,
\[
 \G(A, C)
 \ge \frac{U(C)}{2+e/(e-1)}.
\]
\end{lemma}

Before proving the lemma, we show how to use it to obtain a guarantee for all buyers at once.
Fix $\lambda$ and apply \Cref{lem:non-iid-local} to each buyer, obtaining
mechanisms $A_i$. Write $A=(A_1,\ldots,A_n)$ for this
collection and $U(A):=\sum_iU(A_i)$ for its total expected utility. 

Let $\mathbf q_i^\star$
be buyer $i$'s itemwise ex-ante allocation probabilities in an optimal
multi-buyer mechanism, and let $C_i$ be an optimal single-buyer
mechanism under these caps, attaining
$\SingleOPT(\Vvec_i;\mathbf q_i^\star)$. These mechanisms are used
only as comparators; they need not be computed.

We start with the multi-buyer Lagrangian, decomposing it into per-bidder contribution. We lower-bound the (relaxed) supply-constraint of 1 by the constrained-supply from the optimal single-buyer mechanisms, and the remainder of the term is by definition $\G$. For any single-buyer mechanism's caps, \Cref{lem:non-iid-local} then states that this quantity is a constant-factor of the sum of the utility of these single-buyer mechanisms. The ex-ante relaxation \eqref{eq:ex-ante-relaxation} guarantees that by specifically using the optimal single-buyer mechanisms, the sum of their utilities is at least $\OPT$. 
\begin{align}
 \mathcal{L}_\lambda(A) &\geq 
 \sum_i \G(A_i,C_i) 
       \notag & \text{definition of $\G$}\\
 &\ge\frac{1}{2+e/(e-1)}\sum_i U(C_i)
       \notag & \text{\Cref{lem:non-iid-local}}\\
 &\ge\frac{\OPT}{2+e/(e-1)}.
       & \text{definition of $C_i$, \cref{eq:ex-ante-relaxation}}
 \label{eq:non-iid-price-certificate}
\end{align}

At fixed $\lambda$, lower-bounding $\L_{\lambda}(A)$ does not yet lower-bound $U(A)$:
the supply term $\sum_j\lambda_j(1-\sum_i\eaqj(A_i))$ may be
positive. The projected online gradient descent procedure in \Cref{sec:lambdachoice} for learning $\lambda$ gives a no-regret bound that, on average, controls this extra term. Choosing these $\lambda$s at random over the learning timeline then converts this into a utility approximation. See \Cref{lem:non-iid-balanced}, which also reduces participation to restore ex-ante feasibility.

We now provide the proof of \Cref{lem:non-iid-local}.

\begin{proof}[Proof of \Cref{lem:non-iid-local}]
Fix an arbitrary comparison mechanism $C$ and let its deletion
loss be $D=\sum_jD_j$. The bound from
\Cref{lem:weighted-residual} separates $U(C)$ into deletion loss and a
weighted free-allocation term. We construct one mechanism for each term,
then mix over the two. Note that the
constructed mechanisms need not satisfy the ex-ante probability caps
$\eaqj(C)$ that mechanism $C$ satisfies.

\textbf{Posted Prices $\mpp$.} First, we recover deletion loss through posted prices. To find item-specific prices
with high demand-adjusted utility, consider the following program: 
\begin{equation} \label{eq:11program}
 \max_{\substack{\widehat q_j\ge0\\\sum_j\widehat q_j\le1}}
 \sum_j\left[\sig(V_j;\widehat q_j)-\frac{e}{e-1}\lambda_j\widehat q_j\right].
\end{equation}
Let $\widehat q_j^\star$ denote its optimal caps.
Each summand is concave in $\widehat q_j$. 
Appendix~\ref{app:demand-adjusted-caps} describes a discretized water-filling procedure for approximating this program.  Combining independent prices preserves a $(1-1/e)$ fraction of their single-item utilities while preserving the allocation caps. Multiplying the demand penalty by $e/(e-1)$ ensures that, after applying this utility factor, the penalty is exactly $\lambda_j\widehat q_j$.

Let $\mpp$ be the mechanism that simultaneously offers the (independently drawn) prices
$P_j$ prescribed by $\sig(V_j;\widehat q_j^\star)$ and lets the buyer pick their surplus maximizing item
(or no item). Due to the added unit-demand feasibility of $\mpp$, then $\eaqj(\mpp) \leq \widehat q_j^\star$. 
As in \Cref{sec:iid-single}, we use \Cref{lem:weighted-residual,lem:deletion-coordinate,lem:combine-independent-prices} to recover the two parts of the utility benchmark. 
Its demand-adjusted utility satisfies 
\begin{align}
  \G&(\mpp,C) \notag \\
 &=\E\!\left[\max_j(V_j-P_j)^+\right]
      -\sum_j\lambda_j\eaqj(\mpp)+\sum_j\lambda_j\eaqj(C) \notag 
       & \text{definition of $\G$}\\
 &\ge\left(1-\frac1e\right)\sum_j\sig(V_j;\widehat q_j^\star)
      -\sum_j\lambda_j\widehat q_j^\star+\sum_j\lambda_j\eaqj(C) \notag 
       & \text{\Cref{lem:combine-independent-prices}}, \eaqj(\mpp)\leq\widehat q_j^\star\\
 &\ge\left(1-\frac1e\right)\sum_j\sig(V_j;\eaqj(C)) \notag 
       & \text{ 
       $\eaqj(C)$ feasible for $\widehat q_j^\star$}\\
 &\ge\left(1-\frac1e\right)D.
       & \sig(V_j;\eaqj(C))\ge D_j  
       \label{eq:non-iid-price-bound}
\end{align}

\textbf{Free Lottery $\mfree$.} We next recover the remaining term using free allocation. We seek a free-lottery mechanism $\mfree$ satisfying
\begin{equation}\label{eq:non-iid-free-bound}
 U(C)-D
 \le 2\, \G(\mfree,C). 
\end{equation}
To choose the free lottery mechanism $\mfree$, we introduce parameter $h>0$ hat bounds the expected demand penalty of any offered lottery. For each item
$j$, define probabilities $\ell_j(h):=\min\{1,h/\lambda_j\}$, with $\ell_j(h)=1$ when
$\lambda_j=0$. Let $\mfree_h$ be the free-lottery mechanism in which the
buyer chooses their utility-maximizing item $j$, and receives it for free with probability
$\ell_j(h)$. 
We choose $h$ as the common expected penalty of each lottery, unless giving the item with certainty already incurs less. Increasing $h$ makes each lottery more generous, up to allocation probability 1. This increases the buyer's expected utility, but also allows a larger penalty in our objective $\G$. 
Using \Cref{lem:weighted-residual} we first establish a utility bound for every mechanism $\mfree_h$
\begin{align}
 U(C)-D
 &\le \E_{\Vvec}[\max_j\ell_j(h)V_j]\sum_j\frac{\eaqj(C)}{\ell_j(h)}
       \notag & \text{\Cref{lem:weighted-residual}}\\
 &= U(\mfree_h)\sum_j\frac{\eaqj(C)}{\ell_j(h)}
       \notag & \text{definition of $\mfree_h$}\\
 &\le U(\mfree_h)
       \left(\sum_j\eaqj(C)+\frac{\sum_j\lambda_j\eaqj(C)}{h}\right)
       \notag & \text{$1/\ell_j(h)\le1+\lambda_j/h$}\\
 &\le\frac{U(\mfree_h)}{h}
       \left(h+\sum_j\lambda_j\eaqj(C)\right).
       & \text{$\sum_j\eaqj(C)\le1$}
 \label{eq:non-iid-free-utility}
\end{align}

We would like to select $h$ so that the mechanism's expected utility equals $2h$. Since the buyer chooses only one lottery, the mechanism's expected penalty is at most $h$, leaving penalized utility at least $h$.

The same choice makes the coefficient $U(\mfree_h)/h$ in \eqref{eq:non-iid-free-utility} equal to two. That bound then gives $U(C)-D\le2(h+\sum_j\lambda_j\eaqj(C))$, while subtracting our mechanism's penalty leaves $\G(\mfree_h,C)\ge h+\sum_j\lambda_j\eaqj(C)$. These are the two inequalities we need to combine. The cases below show that we can either choose such an $h$ or use the null mechanism.

We consider two cases.

\emph{Case 1.} $U(\mfree_h)<2h$ for every $h>0$. Then
\eqref{eq:non-iid-free-utility} gives
\[
 U(C)-D<2h+2\sum_j\lambda_j\eaqj(C).
\]
Letting $h$ tend to zero proves \eqref{eq:non-iid-free-bound}
with $\mfree_0$ being the mechanism that never allocates (so $\G(\mfree_0,C)=\sum_j\lambda_j\eaqj(C)$). Set $\mfree=\mfree_0$.

\emph{Case 2.} $U(\mfree_h)\geq2h$ for some $h>0$. Because $U(\mfree_h)/h$ is non-increasing,
continuous and tends to zero as $h$ tends to infinity, we should be able
to find an $h'$ such that $U(\mfree_{h'}) = 2h'$. For $\mfree_{h'}$
we first bound its demand-weight term:
\begin{equation}\label{ineq:non-iid-free-demand-price}
\sum_j\lambda_j \eaqj(\mfree_{h'})
 =\sum_j\lambda_j\ell_j(h')\Prb(\text{buyer chooses item }j)
 \le h'\sum_j\Prb(\text{buyer chooses item }j)
 \le h',
\end{equation}
where the first inequality uses $\lambda_j\ell_j(h')\le h'$ by definition of $\ell_j$,
and the second is because the buyer can only choose one item.

We now lower bound the lottery's demand-adjusted utility:
\begin{align}
 \G(\mfree_{h'},C)
 &=U(\mfree_{h'})-\sum_j\lambda_j\eaqj({\mfree_{h'})}
       +\sum_j\lambda_j\eaqj(C)
       \notag & \text{definition of $\G$}\\
 &\ge U(\mfree_{h'})-h'+\sum_j\lambda_j\eaqj(C)
       \notag & \text{\Cref{ineq:non-iid-free-demand-price}}\\
 &=h'+\sum_j\lambda_j\eaqj(C).
       & \text{$U(\mfree_{h'})=2h'$}
 \label{eq:non-iid-free-adjusted}
\end{align}
Set $\mfree :=\mfree_{h'}$. The general utility bound and the bound on this
lottery's demand-adjusted utility combine as follows:
\begin{align*}
 U(C)-D
 &\le\frac{U(\mfree_{h'})}{h'}
       \left(h'+\sum_j\lambda_j\eaqj(C)\right)
       &&\text{\Cref{eq:non-iid-free-utility}}\\
 &=2\left(h'+\sum_j\lambda_j\eaqj(C)\right)
       &&\text{$U(\mfree_{h'})=2h'$}\\
 &\le 2\, \G(\mfree_{h'}, C)
       &&\text{\Cref{eq:non-iid-free-adjusted}}.
\end{align*}
The final expression is the right-hand side of
\eqref{eq:non-iid-free-bound} for $\mfree =\mfree_{h'}$. In both cases, the
choice of the free-lottery mechanism is independent of $C$.

Having established bounds for the pricing mechanism $\mpp$ and the free lottery mechanism $\mfree$ in 
\Cref{eq:non-iid-price-bound,eq:non-iid-free-bound} we can combine them to get an upper
bound on the utility of mechanism $C$ 
\[
 U(C)
 \le \frac{e}{e-1} \G(\mpp,C) + 2\, \G(\mfree, C).
\]

Let mechanism $\singlehetmech=A$ run the pricing mechanism $\mpp$ with probability $e/(3e-2)$ 
and the free-lottery mechanism $\mfree$ with probability $2(e-1)/(3e-2)$, 
independently of the buyer's report. The last inequality now yields
\[
 U(C)\le\left(2+\frac{e}{e-1}\right) \G(A,C),
\]
which rearranges to the claimed bound, giving the demand-adjusted utility guarantee.

In terms of mechanism $A$'s properties, it uses nonnegative
payments and is ex-post individually rational (since both of its sub-mechanisms do so as well), 
and since the buyer chooses their preferred offer in each sub-mechanism, the mixture is BIC as well.
\end{proof}

To recap, our single-buyer mechanism $\singlehetmech$ is defined as follows: Fix $\lambda\ge0$. Let $\widehat{\mathbf q}^\star$ maximize the cap program \eqref{eq:11program} in the proof
of \Cref{lem:non-iid-local}. Then

\begin{enumerate}
 \item With probability $e/(3e-2)$, run $\mpp$: Draw each item $j$'s price independently
 according to $\sig(V_j;\widehat q_j^\star)$. Let the buyer choose the item with the largest positive utility (if any) and charge its price.
 \item With probability $2(e-1)/(3e-2)$, run $\mfree$: Let the buyer choose their utility-maximizing item $j$
 and allocate it to them with probability $\ell_j(h)$ free of charge. Here $\ell_j(h)=\min\{1,h/\lambda_j\}$, with $\ell_j(h)=1$ when $\lambda_j=0$. Choose $h>0$ satisfying $\E[\max_j\ell_j(h)V_j]=2h$, if one exists; otherwise allocate nothing.
\end{enumerate}

We can then conclude that the better (demand-adjusted-utility-maximizing) of these two mechanisms gives a constant-factor approximation.

\begin{corollary}
Under the assumptions of \Cref{lem:non-iid-local}, choosing
$A\in\{\mpp,\mfree\}$ to maximize
$U(A)-\sum_j\lambda_j\eaqj(A)$ gives
\[
 \G(A,C)\ge\frac{U(C)}{2+e/(e-1)}
\]
for every comparison mechanism $C$.
\end{corollary}

\subsubsection{Choosing $\lambda$'s.} \label{sec:lambdachoice} 
We now show how we can construct an ex-ante-feasible approximately-optimal
multi-buyer mechanism $\indmech$ by learning $\lambda$.  We initialize the dual variables $\lambda$ at 0, and then iteratively update them for $T$ rounds. When an item is over-allocated, they are increased, and if under-allocated, decreased. They are constrained to be nonnegative and to never exceed a crude upper bound of welfare if allocation of all $m$ items to each of $n$ buyers was possible. 

The standard no-regret learning guarantee of projected online gradient descent \citep{Zinkevich03} is in average over the $T$ rounds. Hence, we randomly draw a mechanism from the last $T$ rounds to get the same expected guarantee.

Formally, initialize demand weights to $\lambda^1=0$ and set a 
step size $\eta = \Theta(S/(n\sqrt{T}))$ (where $S = \sum_{i,j}\E[V_{ij}]$). We will run the following procedure
for $T$ rounds. In round $t$, apply \Cref{lem:non-iid-local} with
demand weights $\lambda^t$ to compute a mechanism
$A_i^t:=\mathcal M_{\mathrm{het},i}^{1,m}(\lambda^t)$ for each buyer $i$. Let $r_j^t:=\sum_i\eaqj(A_i^t)$ be the
expected number of requests for item $j$ in this round, and update
its the demand weights according to
\begin{equation}\label{eq:non-iid-price-update}
 \lambda_j^{t+1}
 =\min\left\{\sum_{i,j}\E[V_{ij}],\max\left\{0,\lambda_j^t+\eta(r_j^t-1)\right\}\right\}.
\end{equation}
After the
$T$ rounds, define the average expected number of requests 
$\overline r_j:=T^{-1}\sum_{t=1}^T r_j^t$.
Mechanism $\indmech$ independently draws a uniform round $t_i$ for
each buyer $i$ and runs for this buyer $A_i^{t_i}$ with probability
$1/\max\{1,\max_j\overline r_j\}$ or otherwise makes no offer to
that buyer. The following lemma describes the properties and utilities
guarantees of mechanism $\indmech$ (the proof of the following lemma and exact 
implementation details of $\indmech$ are deferred to 
Appendix~\ref{app:non-iid-average-proof}).

\begin{restatable}{lemma}{noniidbalanced}\label{lem:non-iid-balanced}
Using the demand weights $\lambda^1,\ldots,\lambda^T$
generated by the procedure above, for every $\varepsilon>0$, after 
$T=O(n^4m^4/\varepsilon^2)$ rounds, $\indmech$ is
ex-ante feasible and the mechanism's total expected utility satisfies
\[
 U(\indmech)
 \ge\left(\frac{1}{2+e/(e-1)}-\varepsilon\right)\OPT.
\]
\end{restatable}

\subsection{Resolving contention without changing buyers' choices} \label{sec:hetservice}

Let $\indmech$ be the mechanism constructed by the demand-weight
procedure in \Cref{sec:lambdachoice}. By \Cref{lem:non-iid-balanced},
its requests satisfy each item feasibility constraint in expectation,
but several buyers may still request the same item. In
\Cref{sec:service}, item symmetry ensured that every item a buyer
might request had the same probability of being served. Here,
competition can differ across items. To preserve the buyer's choices,
we first resolve contention using a constant lower bound on each
request's service probability, and then further reduce these probabilities to take on the same value.

Write $p_{ij}$ for the probability that buyer $i$ requests item $j$
under $\indmech$. For each item $j$, requests are independent across
buyers and $\sum_i p_{ij}\le1$. We apply the rank-one contention rule
of Kashaev and Santiago~\cite[Algorithm~2.1 and Theorem~2.1]{KashaevSantiago23}
to these requests (we could have alternatively used the rank-one rule of \cite{FV06}). This rule selects at most one buyer (item feasibility), and each
request from buyer $i$ is selected with probability
$c_{ij}\ge1-1/e$. 
When $p_{ij}=0$, define $c_{ij}$ by forcing buyer $i$ to make a request, and then
run the same rule against the other buyers' requests. The bound also covers requests due to misreports. 
Appendix~\ref{app:non-iid-probabilities} explains how to
compute $p_{ij}$ and $c_{ij}$.

If the rule selects buyer $i$ for item $j$, serve the request it with probability
$(1-1/e)/c_{ij}$, independently of the earlier randomness. Charge
the corresponding single-buyer payment only if the item is served.
Then every possible request is served with probability exactly
$1-1/e$, when holding the buyer's values, report, and single-buyer
randomness fixed. Both the allocation and payment are multiplied by
this same factor, so the buyer's preference between item options is
unchanged. By \Cref{lem:utility-preserving-contention}, the resulting
mechanism $\hetmech$ is BIC, ex-post individually rational, and
matching-feasible, has nonnegative payments, and satisfies
\begin{equation}\label{eq:non-iid-contention-utility}
 U(\hetmech)=\left(1-\frac1e\right)U(\indmech).
\end{equation}
Combining this utility guarantee with \Cref{lem:non-iid-balanced}
gives the following approximation.

\begin{restatable}{theorem}{independentbuyers}\label{thm:non-iid-multi}
In the independent-values setting, for every $\varepsilon>0$,
the demand-weight procedure followed by the contention rule above
gives a BIC, ex-post individually rational, matching-feasible mechanism
$\hetmech$ with nonnegative payments and expected utility at least
\[
 U(\hetmech) \geq \frac{\OPT}{2e/(e-1)+(e/(e-1))^2+\varepsilon}.
\]
The mechanism can be computed in time polynomial in
$n,m,1/\varepsilon$, given access to the single-buyer optimization
and expectation evaluations used in this section.
\end{restatable}

\begin{proof}
It suffices to consider $0<\varepsilon\le1$; for larger
$\varepsilon$, use the mechanism obtained with $\varepsilon=1$.
We choose the accuracy in \Cref{lem:non-iid-balanced} so that,
after the utility loss due to contention, the additive error in the
approximation factor is $\varepsilon$. Set
\[
 \theta:=\frac{1}{2+e/(e-1)}
 -\frac{1}{2+e/(e-1)+(1-1/e)\varepsilon}>0.
\]
Running the demand-weight procedure with accuracy $\theta$ gives
\[
 U(\indmech)
 \ge\left(\frac{1}{2+e/(e-1)}-\theta\right)\OPT
 =\frac{\OPT}{2+e/(e-1)+(1-1/e)\varepsilon}.
\]
By \eqref{eq:non-iid-contention-utility},
\[
 U(\hetmech)
 \ge\frac{1-1/e}{2+e/(e-1)+(1-1/e)\varepsilon}\OPT
 =\frac{\OPT}{2e/(e-1)+(e/(e-1))^2+\varepsilon}.
\]
Since $\theta=\Theta(\varepsilon)$ for $0<\varepsilon\le1$,
the construction uses polynomially many rounds and single-buyer
calls. The contention rule and the computation of its selection probabilities also take polynomial time under the stated access
assumptions.
\end{proof}

\subsection*{AI Disclosure}

We used ChatGPT extensively in developing this paper. This work built on five years of our preliminary research, including results, technical approaches, and intuition about what had and had not worked. We supplied this background and guided an iterative investigation: we formulated the original research questions, posed further questions as the work developed, suggested techniques, and directed which approaches to pursue. Within this process, ChatGPT generated a draft paper and proofs of most of the results, primarily using GPT-6 Astra. The first proof of the result for independent values was obtained using GPT-5.6. 

We wrote the paper and proofs presented here by completely rewriting the AI-generated drafts. This involved reinterpreting and simplifying the arguments, developing new notation, breaking the proofs into smaller conceptual steps, and drawing out the underlying intuition. ChatGPT contributed substantially to discovering results and generating proofs within that process. We take responsibility for the paper's content.

\nocite{BoydSubgradient}
\bibliographystyle{plainnat}
\bibliography{masterbib}

\clearpage
\appendix

\section{Stronger favorite-item guarantees}
\label{sec:combined-bounds}

We prove a factor-$2$ single-buyer guarantee for the \emph{favorite-item}
mechanism, then combine it with contention resolution for fully i.i.d.~and item-i.i.d.~buyers.
Throughout this appendix, $Q$ denotes a buyer's total
allocation cap, $q=Q/m$ the common item cap, and
$\mathbf q=(q,\ldots,q)$ the vector of item caps. By item symmetry,
the optimum under total cap $Q$ equals
$\SingleOPT(\Vvec;\mathbf q)$, using the notation of the main body.

\subsection{A factor-2 approximation through the favorite-item mechanism}
\label{app:improved-single}

The simplified mechanism offers one random price, common to all items, according to the \emph{favorite-item 
mechanism} $\sig(M;Q)$, where
$M:=\max_j V_j$. The buyer selects a favorite item if its value strictly
exceeds the price. By symmetry, each item's allocation
probability is at most $Q/m$.
To prove \Cref{lem:favorite-two}, we use the prices from $\sig(V;q)$ to construct a mechanism for
the favorite value $M$, recovering the single-item utility terms of
\Cref{thm:single-buyer-benchmark} without the loss from combining
independent item prices. The following lemma states the resulting approximation guarantee.

\begin{restatable}{lemma}{favoritetwo}\label{lem:favorite-two}
For i.i.d.~item values $\Vvec$ and every total ex-ante cap $Q\in[0,1]$,
with $q=Q/m$ and $\mathbf q=(q,\ldots,q)$, the optimal 
favorite-item  mechanism $\sig(M;Q)$ is a 2-approximation to the optimal single-buyer
mechanism $\SingleOPT(\Vvec;\mathbf q)$, that is
\[
 \SingleOPT(\Vvec;\mathbf q)\le2\sig(M;Q).
\]
\end{restatable}

Fix a single buyer with i.i.d.~item values, a total cap $Q\in[0,1]$,
and a comparison mechanism with utility $U$ and ex-ante allocation
probability at most $Q$. As before, it is without loss of generality to consider mechanisms with buyer item ex-ante probabilities of $Q/m$. From 
\Cref{lem:weighted-residual,lem:deletion-coordinate} with $\ell_j = Q/m$ we get 
\begin{equation}\label{eq:appendix-decomposition}
 U\le D+Q\E[M]\le m\sig(V;Q/m)+Q\E[M].
\end{equation}
We first recover the deletion bound through a mechanism for the
favorite value $M$, using a smaller allocation cap.

\begin{lemma}\label{lem:favorite-compression}
Let $V,V_1,\ldots,V_m$ be i.i.d.~nonnegative random variables drawn
from $F$, and let $M:=\max_jV_j$. For every $Q\in[0,1]$,
\begin{equation}\label{eq:parent-compression}
 m\sig(V;Q/m)\le(1+Q)\sig\left(M;\frac{Q}{1+Q}\right).
\end{equation}
\end{lemma}

\begin{proof}[Proof of \Cref{lem:favorite-compression}]
Fix an optimal $(Q/m)$-constrained single-dimensional mechanism for
$V$, whose expected utility is $\sig(V;Q/m)$. Since this is a
single-dimensional mechanism, we can view it as a distribution over
posted-price mechanisms; for details we refer to \cite{HartlineBook}. For a single-dimensional value $X$, a
posted price $t$ is accepted with probability $\Prb(X>t)$ and provides
expected utility $\E[(X-t)^+]$. (Accepting when 
$X=t$ contributes no utility and may therefore be omitted.)

Let $\omega_k$ and $t_k$ denote the weights and prices in the
distribution over posted prices representing $\sig(V;Q/m)$. For each
component $k$, let $\pi_k^V:=\Prb(V>t_k)$ denote the probability that its
posted price is accepted. Feasibility requires that the weights define
a valid distribution and that the ex-ante allocation probability under
$V$ is at most $Q/m$, that is,
\[
  \sum_k\omega_k\le1,
  \qquad
  \sum_k\omega_k \pi_k^V\le Q/m.
\]
We will construct a valid distribution over posted-price mechanisms
for the maximum value $M$ that obtains a fraction of the utility of
the mechanism for $V$. We use the same prices $t_k$ and let
$\pi_k^M:=\Prb(M>t_k)$ denote their acceptance probabilities under $M$.
For each component $k$ with $\pi_k^V>0$, define its probability weight by
\[
  \widetilde\omega_k
  :=\frac{m\omega_k\pi_k^V}{(1+Q)\pi_k^M}.
\]

We first verify that the weights $\widetilde\omega_k$ define a
valid distribution. Independence gives
\begin{equation}
  \pi_k^M
  =1-(1-\pi_k^V)^m
  \ge 1-e^{-m\pi_k^V}
  \ge 1-\frac{1}{1+m\pi_k^V}
  =\frac{m\pi_k^V}{1+m\pi_k^V},
  \label{eq:max-acceptance-lower-bound}
\end{equation}
where both inequalities follow from $e^x\ge1+x$, applied with
$x=-\pi_k^V$ and $x=m\pi_k^V$, respectively. We can therefore upper-bound
the sum of the new probability weights:
\[
  \sum_k\widetilde\omega_k
  =\sum_k\frac{m\omega_k\pi_k^V}{(1+Q)\pi_k^M}
  \le\frac{\sum_k\omega_k(1+m\pi_k^V)}{1+Q}
  \le\frac{1+Q}{1+Q}
  =1,
\]
where the first inequality is by
\eqref{eq:max-acceptance-lower-bound}, while the second follows from
the feasibility of the original distribution, namely
$\sum_k\omega_k\le1$ and $\sum_k\omega_k\pi_k^V\le Q/m$.
Thus, the weights define a valid distribution (with any remaining
probability assigned to making no offer).

The new mechanism's ex-ante allocation probability is
\[
  \sum_k\widetilde\omega_k\pi_k^M
  =\sum_k\frac{m\omega_k\pi_k^V}{(1+Q)\pi_k^M}\pi_k^M
  =\frac{m}{1+Q}\sum_k\omega_k\pi_k^V
  \le\frac{Q}{1+Q},
\]
where the inequality follows from the ex-ante allocation probability
of the original mechanism, $\sum_k\omega_k\pi_k^V\le Q/m$.

To compare utilities after reweighting, we show that an accepted
offer gives at least as much expected surplus for $M$ as for $V$.
Fix a price $t$ with $\Prb(V>t)>0$.
For $s\ge t$ with $\Prb(V>s)>0$,
\[
 \frac{\Prb(M>s)}{\Prb(V>s)}
 =\frac{1-F(s)^m}{1-F(s)}
 =\sum_{\ell=0}^{m-1}F(s)^\ell
 \ge\sum_{\ell=0}^{m-1}F(t)^\ell
 =\frac{\Prb(M>t)}{\Prb(V>t)}.
\]
Rearranging the tail inequality and integrating over $s\ge t$
yields
\begin{equation}\label{ineq:max-single-comp}
 \frac{\E[(M-t)^+]}{\Prb(M>t)}
 \ge\frac{\E[(V-t)^+]}{\Prb(V>t)}.    
\end{equation}

Let $U_M$ denote the expected utility of the new mechanism. We lower bound its utility
\[
  \begin{aligned}
    U_M
    &:=\sum_k\widetilde\omega_k\E[(M-t_k)^+]
      =\frac{m}{1+Q}\sum_k\omega_k\frac{\pi_k^V}{\pi_k^M}\E[(M-t_k)^+]\\
    &\ge\frac{m}{1+Q}\sum_k\omega_k\E[(V-t_k)^+]
      =\frac{m}{1+Q}\sig(V;Q/m),
  \end{aligned}
\]
where the inequality is due to~\eqref{ineq:max-single-comp}.

The constructed mechanism is therefore feasible for $\sig(M;Q/(1+Q))$. By optimality and the
preceding utility bound,
\[
  \sig(M;Q/(1+Q))\ge U_M\ge\frac{m}{1+Q}\sig(V;Q/m),
\]
which rearranges to prove the claim.
\end{proof}

We can now apply this to inequality~\eqref{eq:appendix-decomposition} to prove \Cref{lem:favorite-two}:

\begin{proof}[Proof of \Cref{lem:favorite-two}]
Consider $0<Q\le1/2$. We construct a single buyer mechanism
for $M$ with ex-ante probability $Q$ by mixing the two mechanisms corresponding to the terms of
\eqref{eq:appendix-decomposition}. With probability $(1+Q)/2$, run
$\sig(M;Q/(1+Q))$; with probability $Q/2$, allocate the favorite item
for free; and with the remaining probability $(1-2Q)/2$, offer nothing. The mixture's ex-ante allocation
probability is at most
\[
 \frac{1+Q}{2}\frac{Q}{1+Q}+\frac Q2=Q.
\]
By \Cref{lem:favorite-compression}, its expected utility is at least
\begin{equation}\label{eq:combined-q-bound}
 \frac{1+Q}{2}\sig\left(M;\frac{Q}{1+Q}\right)
 +\frac Q2\E[M]
 \ge\frac12\bigl(m\sig(V;Q/m)+Q\E[M]\bigr)
 \ge\frac U2,
\end{equation}
where the final inequality is \eqref{eq:appendix-decomposition}. The
mixture preserves incentive compatibility, individual rationality,
and nonnegative payments, and is therefore feasible for
$\sig(M;Q)$. Hence $U\le2\sig(M;Q)$ in this case.

For $1/2<Q\le1$, nonnegative payments
and unit-demand feasibility imply $U\le\E[M]$, while giving the
favorite item for free with probability $Q$ shows that
$\sig(M;Q)\ge Q\E[M]\ge U/2$. 

At $Q=0$, both the comparison mechanism
and $\sig(M;Q)$ have zero utility. 

Since the comparison mechanism was
arbitrary, we can apply this result to a mechanism with utility $\SingleOPT(\Vvec;\mathbf q)$, which means we have proven
\begin{equation}\label{eq:combined-single-bound}
 \SingleOPT(\Vvec;\mathbf q)\le2\sig(M;Q).
\end{equation}
\end{proof}

\subsection{Fully i.i.d.~buyers}
\label{app:iid-favorite}

Using \Cref{lem:favorite-two}, we can immediately improve on the
approximation factor of \Cref{thm:iid-simple} by a factor of
$(2e-1)/(2(e-1))$.

\begin{restatable}{theorem}{iidfavorite}\label{thm:iid-favorite}
For i.i.d.~values, independently running
$\sig(M;\min\{1,m/n\})$ for each buyer and resolving contention uniformly at
random gives a BIC, ex-post individually rational, matching-feasible mechanism with
nonnegative payments, and achieves expected utility at least
\[
 \frac{1-1/e}{2}\OPT.
\]

\end{restatable}
This follows by applying \Cref{lem:favorite-two} instead of \Cref{lem:single}
in the analysis of Section~\ref{sec:service}.

\subsection{Different distributions across buyers}
\label{sec:buyer-specific}

We now extend the guarantee to the \emph{item-i.i.d.}~setting:
buyer $i$ has i.i.d.~item values from $F_i$, but these distributions may
differ across buyers, although they remain independent. Write
$M_i:=\max_jV_{ij}$. Both single-buyer guarantees (\Cref{sec:iid-single}, \Cref{lem:favorite-two}) apply separately to
each buyer, with total cap $Q_i$ and item caps $q_{ij}=Q_i/m$.
The multi-buyer argument changes in two respects.

The first issue pertains to computing appropriate ex-ante total probability caps $Q_i$. 
If the number of items is larger than the number of buyers ($m\ge n$), setting $Q_i=1$ 
for every buyer remains feasible and optimal, so we restrict to cases with fewer items ($m<n$).
Without symmetry across buyers, the common cap $Q=m/n$ may not be optimal, and we therefore need buyer-specific caps $Q_i$.

Ideally, we would choose these caps by optimizing the ex-ante relaxation itself. Even though its objective is concave in the allocation caps, we cannot efficiently optimize it directly because we cannot efficiently evaluate feasible solutions and compute their gradients. Instead, we use \Cref{lem:favorite-two} to upper-bound the relevant single-buyer ex-ante objectives by the simpler 1-dimensional objective $2\sig(M_i;Q_i)$, and optimize

\begin{equation}\label{eq:ex-ante-favorite}
 \mathcal{R}:=\max_{\substack{0\le Q_i\le1\\\sum_iQ_i\le m}}
       \sum_{i=1}^n\sig(M_i;Q_i).
\end{equation}
For the single-buyer interim mechanisms induced by an optimal
multi-buyer mechanism, let $Q_i^C$ be buyer $i$'s ex-ante allocation
probability. Since $Q_i^C\le1$ and $\sum_iQ_i^C\le m$, applying
\eqref{eq:combined-single-bound} to each buyer gives
\begin{equation}\label{eq:heterogeneous-upper}
 \OPT\le2\mathcal{R}.
\end{equation}
The advantage of $\mathcal{R}$ is that each $\sig(M_i;Q_i)$ can be evaluated efficiently. Moreover, the objective is the sum of concave functions, each depending only on one buyer's cap. We can therefore use a discretized greedy water-filling procedure to find approximately optimal caps. In particular, we can obtain caps whose objective value is within a multiplicative factor $1+\varepsilon$ of $\mathcal{R}$ in time $\operatorname{poly}(n,m,1/\varepsilon)$ (see Appendix~\ref{app:greedy-caps} and 
\Cref{lem:greedy-caps} for the details).

The second issue is contention resolution. Unlike the
fully i.i.d.~setting, some buyers may contribute much more utility
than others (consider for example, a buyer whose item values are always
much higher than everyone else's) so treating all requests
uniformly may not retain the same utility guarantee. Fortunately,
the contention construction of \Cref{sec:hetservice} applies here
as well: it serves every possible request with probability
$1-1/e$ while preserving buyers' choices.

The following theorem extends the favorite-item guarantee to
item-i.i.d.~values.

\begin{restatable}{theorem}{iidstronger}
\label{thm:iid-stronger}
For item-i.i.d.~values and every $0 < \varepsilon \leq 1$, we can compute
buyer allocation caps $Q_i\in[0,1]$ with $\sum_iQ_i\le m$ and
a contention-resolution scheme such that independently running
$\sig(M_i;Q_i)$ for each buyer, where $M_i:=\max_jV_{ij}$,
and resolving the resulting requests using this scheme gives a
mechanism that is BIC, ex-post individually rational, has
nonnegative payments, respects matching-feasibility, and achieves
expected utility at least
\[
    \frac{\OPT}{2e/(e-1)+\varepsilon}.
    \]
The mechanism can be computed in $\operatorname{poly}(n,m,1/\varepsilon)$
time under the single-dimensional optimization and probability-evaluation
assumptions of Appendix~\ref{app:optimization}.
    
\end{restatable}

\begin{proof}[Proof of \Cref{thm:iid-stronger}]
If $m\ge n$, set $Q_i=1$ for every buyer---these are optimal ex-ante caps. 
Otherwise, choose
$\delta:=\varepsilon(1-1/e)/2$ and use the
item-i.i.d.~cap procedure in \Cref{lem:greedy-caps} to compute
feasible caps satisfying
$\sum_i\sig(M_i;Q_i)\ge\mathcal R/(1+\delta)$.
For the following bound, use the same definition of $\delta$
in both cases.

Independently run $\sig(M_i;Q_i)$ for each
buyer, breaking item ties uniformly. Item symmetry bounds buyer
$i$'s probability of requesting any fixed item by $Q_i/m$, so
the requests are independent across buyers and have total
expected demand at most $\sum_iQ_i/m\le1$ for every item.
The contention construction of \Cref{sec:hetservice} therefore
applies directly: it preserves the stated incentive, payment,
and feasibility properties and retains a $1-1/e$ fraction of
the single-buyer utilities. Expected utility is thus at least
\[
 \left(1-\frac1e\right)\sum_i\sig(M_i;Q_i)
 \ge\frac{1-1/e}{1+\delta}\mathcal R
 \ge\frac{1-1/e}{2(1+\delta)}\OPT
 \ge\frac{\OPT}{2e/(e-1)+\varepsilon},
\]
where the inequalities use the computed caps' objective guarantee,
the ex-ante upper bound
\eqref{eq:heterogeneous-upper}, and the choice of $\delta$,
respectively. The computational guarantees follow from the cap
procedure, the single-dimensional optimization, and contention
probability calculations in Appendix~\ref{app:optimization}.
\end{proof}

\section{Approximate optimization and further implementation details}
\label{app:optimization}

This appendix provides the optimization and implementation
details used by our mechanisms. We first relate the computation of
$\sig$ to standard single-dimensional optimization and specify the
computational access used here (Appendix~\ref{sec:compute}).
We then give a shared discretized greedy procedure for choosing
item-i.i.d.~caps and demand-adjusted caps, with the approximation
guarantees and number of calls to the single-dimensional optimization
procedure for each application (Appendix~\ref{app:greedy-caps}).
Next, we prove that averaging the mechanisms from the demand-weight rounds
and reducing how often they run gives ex-ante feasible mechanisms with
the utility bound in \Cref{lem:non-iid-balanced} (Appendix~\ref{app:non-iid-average-proof}).
Finally, we explain how to compute request probabilities across these
rounds and the selection probabilities used to resolve contention
(Appendix~\ref{app:non-iid-probabilities}).
In the optimization procedures, we count the number of calls to the single-dimensional optimization procedure
separately from the operations used to choose caps.

\subsection{Solving the single-dimensional optimization problem}\label{sec:compute}

We formulate the single-dimensional problem directly in terms
of allocation and payment under posted prices. At a finite price $t\ge0$,
the buyer receives the item and pays $t$ exactly when $X>t$. Thus the
allocation and payment rules are
\[
  a_t(z):=\mathbf 1_{\{z>t\}},
 \qquad p_t(z):=t\mathbf 1_{\{z>t\}}.
\]
Rejecting at equality contributes no utility and can only
reduce the allocation. We use $t=\infty$ to denote no offer, with
$a_\infty=p_\infty=0$. Using the posted-price representation of a
single-dimensional mechanism, we optimize over probability distributions
$\mu$ on these prices:
\[
    \sig(X;q)
 =\max_{\mu}\left\{
   \int \E_X[Xa_t(X)-p_t(X)]\,d\mu(t):
   \int \E_X[a_t(X)]\,d\mu(t)\le q,\;\;
   \mu\text{ is a distribution}
 \right\}.
\]

Previous works have already characterized the optimal solution to be a mixture
over at most two prices (possibly including a price of $0$ or $\infty$), and
have described methodology on how to solve this type of optimization program \citep{CMM11,AFHH13,Alaei11}.
We refer to \citet{Alaei11} (Theorem 9 and Lemma 2) for details
on how to efficiently compute the concave closures required
in ironing, and how to efficiently find approximately-optimal mechanisms.
Following the standard computational assumptions in Bayesian
mechanism design, we assume efficient access to the ironed
virtual-value functions and the threshold and probability
calculations required by this procedure.

\subsection{Discretized greedy cap optimization}
\label{app:greedy-caps}
\label{app:demand-adjusted-caps}

The item-i.i.d.~cap program \eqref{eq:ex-ante-favorite} and the demand-adjusted program
\eqref{eq:11program} can be approximated by the same greedy discretized procedure.
Both maximize a sum of concave functions, each depending on one
cap, subject to a bound on the sum of the caps. We give the
procedure and its discretized optimality argument once, then
derive the accuracy guarantees for the two applications.

Let $\phi_1,\ldots,\phi_d$ be concave functions on $[0,1]$
with $\phi_i(0)=0$, and define the concave objective
$\Phi(q):=\sum_{i=1}^d\phi_i(q_i)$ for $q\in[0,1]^d$.
For an integer total cap $C\in\{1,\ldots,d\}$, consider
\begin{equation}\label{eq:concave-program}
 \max_{\substack{0\le q_i\le1\\\sum_iq_i\le C}}
\Phi(q).  
\end{equation}

We use the following greedy algorithm.
Fix an integer $K\ge1$ and use increments of $1/K$.
The integer $k_i$ records the increments assigned to coordinate $i$.
Define the gain of coordinate $i$'s $k$th increment by
$g_{ik}:=\phi_i(k/K)-\phi_i((k-1)/K)$ for $k=1,\ldots,K$.
The algorithm maintains $g_i$, the gain of coordinate $i$'s next
available increment.
\begin{enumerate}
 \item Set $k_i=0$ and $g_i=g_{i1}$ for every $i$.
 \item While $\sum_i k_i<CK$:
 \begin{enumerate}
  \item Choose a coordinate $i$ maximizing $g_i$, breaking ties
  arbitrarily. If $g_i\le0$, stop.
  \item Increase $k_i$ by 1. If $k_i=K$, set $g_i=-\infty$.
  Otherwise, update $g_i=g_{i,k_i+1}$.
 \end{enumerate}
 \item Return $q_i=k_i/K$ for every $i$.
\end{enumerate}
The returned caps satisfy $q_i\le1$ and $\sum_iq_i\le C$.
We stop when the largest available gain is nonpositive, since 
no remaining increment can improve the objective. Note that 
the procedure uses at most $d+CK$ evaluations of the
functions $\phi_i$: $d$ for initialization and at most one per
increment, since only the selected coordinate's next gain changes.
Selecting the largest available gain requires only
polynomial additional time in $d$ and $CK$.

We now analyze its approximation guarantees.
Let $q^{\mathrm G}$ be the vector returned by the actual greedy
algorithm. Denote its objective value by
$\Phi_{\mathrm G}:=\Phi(q^{\mathrm G})$. The optimal
value of the discretized problem is
\[
 \Phi_K^\star
 :=\max_{\substack{k_i\in\{0,\ldots,K\}\\\sum_i k_i\le CK}}
       \Phi(k_1/K,\ldots,k_d/K).
\]
We first prove that the greedy algorithm solves the discretized
problem optimally, that is, $\Phi_{\mathrm G}=\Phi_K^\star$.
Using the
increment gains defined in the algorithm, concavity gives
$g_{i1}\ge\cdots\ge g_{iK}$, so the largest available gain is
also the largest of all unselected gains.
Any optimal solution to the discretized problem decomposes
into at most $CK$ increments, whose gains we can order nonincreasingly.
A greedy exchange argument then shows that the greedy algorithm attains
at least the same objective value, so $\Phi_{\mathrm G}=\Phi_K^\star$.

We now compare the greedy objective $\Phi_{\mathrm G}$ with
the continuous optimal value $\Phi^\star:=\Phi(q^\star)$, where
$q^\star$ is an optimal continuous cap vector. For this comparison,
round each cap upward to $q_i^+:=\lceil Kq_i^\star\rceil/K$ and
denote the rounded vector's objective by $\Phi_+:=\Phi(q^+)$.
This vector remains in $[0,1]^d$, but
its number of increments satisfies
\begin{equation}\label{eq:greedy-rounded-increments}
 \sum_i Kq_i^+
 =\sum_i\lceil Kq_i^\star\rceil
 \le K\sum_i q_i^\star+d
 \le CK+d,
\end{equation}
where the first inequality uses $\lceil x\rceil\le x+1$,
and the second uses feasibility of $q^\star$. Thus $q^+$ may exceed the actual
algorithm's limit of $CK$ increments.

To bound the benefit of these additional $d$ increments, we define an \emph{extended
greedy run}, used only in the proof. It follows the greedy
algorithm's selection rules, assigning at
most $K$ increments to each coordinate, but omits the total
increment limit of $CK$ and stops when no positive gain remains. The extra
increments do not need to be assigned one per coordinate; the extended
run selects whichever remaining gains are largest. 

Let $\gamma_t$ be the gain selected at step $t$ of the extended
run, setting $\gamma_t=0$ after it stops. The greedy argument used to establish optimality
of $\Phi_{\mathrm G}$ also bounds the rounded vector's objective:
\begin{equation}
 \begin{aligned}
 \Phi_+
 &\le\sum_{t=1}^{CK+d}\gamma_t
 =\sum_{t=1}^{CK}\gamma_t+\sum_{t=CK+1}^{CK+d}\gamma_t\\
 &\le\sum_{t=1}^{CK}\gamma_t+d\gamma_{CK}
 \le\left(1+\frac{d}{CK}\right)\sum_{t=1}^{CK}\gamma_t
 =\left(1+\frac{d}{CK}\right)\Phi_{\mathrm G}.
 \end{aligned}
 \label{eq:greedy-prefix-comparison}
\end{equation}
where the first inequality follows from the greedy argument,
the second inequality bounds each additional gain by
$\gamma_{CK}$, and the third uses
$\gamma_{CK}\le (\sum_{t=1}^{CK}\gamma_t) / (CK)$,
since the gains are nonnegative and nonincreasing, and the
final equality holds because the greedy algorithm and the extended
run make the same choices until the greedy algorithm stops.

It remains to relate $\Phi_+$ to the continuous optimum
$\Phi^\star$. We account for the effect of rounding
separately in each application below.

\paragraph{Item-i.i.d.~optimization.}
If $m\ge n$, setting every $Q_i=1$ is optimal by monotonicity.
Otherwise, with $d=n$, $C=m$, and $\phi_i(q)=\sig(M_i;q)$,
the general concave program \eqref{eq:concave-program}
coincides with the item-i.i.d.~cap program
\eqref{eq:ex-ante-favorite}, whose optimal value is $\mathcal R$.

For $0<\varepsilon\le1$, choose
$K:=\lceil n/(m\varepsilon)\rceil$ and apply the preceding
greedy algorithm.
It returns caps $Q_i=k_i/K$ with $\sum_iQ_i\le m$.

\begin{lemma}\label{lem:greedy-caps}
For $m<n$ and $0<\varepsilon\le1$, the greedy algorithm
runs in $\operatorname{poly}(n,m,1/\varepsilon)$ time, makes
$\operatorname{poly}(n,1/\varepsilon)$ calls to the single-dimensional
optimization procedure, and returns feasible caps $Q_i$ satisfying
\[
 \sum_i\sig(M_i;Q_i)\ge\frac{\mathcal R}{1+\varepsilon}.
\]
\end{lemma}

\begin{proof}
With $d=n$, $C=m$, and $\phi_i(q)=\sig(M_i;q)$, we have
$\Phi^\star=\mathcal R$ and
$\Phi_{\mathrm G}=\sum_i\sig(M_i;Q_i)$.
For the rounded optimal caps $q_i^+=\lceil Kq_i^\star\rceil/K$,
we have $q_i^+\ge q_i^\star$. Combining monotonicity with
the greedy comparison gives
\begin{align*}
 \mathcal R
 &=\Phi^\star\le\Phi_+\\
 &\le\left(1+\frac{n}{mK}\right)\Phi_{\mathrm G}\\
 &\le(1+\varepsilon)\sum_i\sig(M_i;Q_i),
\end{align*}
where the first inequality follows from rounding the optimal caps upward and monotonicity of $\sig(M_i;\cdot)$,
the second is \eqref{eq:greedy-prefix-comparison} with $C=m$
and $d=n$, and the third uses $mK\ge n/\varepsilon$.

Finally, the greedy algorithm uses at most
$n+mK=O(n/\varepsilon)$ function evaluations, with
polynomial running time as established for the general
procedure, giving the stated computational guarantees.

\end{proof}
\paragraph{Demand-adjusted optimization.}

We next apply the greedy algorithm to the demand-adjusted
optimization problem from \eqref{eq:11program}. This problem
chooses itemwise allocation caps with total at most 1 to
maximize the sum of the optimal single-dimensional utilities,
minus their demand penalties. In the general continuous concave
program \eqref{eq:concave-program}, set $d=m$, $C=1$, and
\[
 \phi_j(q):=\sig(V_j;q)-\frac{e}{e-1}\lambda_jq.
\]
With these choices, the general program is precisely the
demand-adjusted optimization problem. Denote its optimal value by
$\mathcal R_\lambda$, meaning that $\Phi^\star=\mathcal R_\lambda$,
and let $\widehat q_j^\star$ be optimal caps.

For $0<\delta\le1$, choose
$K:=\lceil m(1+e/(e-1))/\delta\rceil$ and apply the greedy algorithm.
The following lemma gives an additive guarantee for this optimization.

\begin{lemma}\label{lem:greedy-demand-adjusted}
Let $S:=\sum_{i,j}\E[V_{ij}]>0$ and
$0\le\lambda_j\le S$ for every item $j$.
For $0<\delta\le1$, the greedy algorithm runs in
$\operatorname{poly}(m,1/\delta)$ time, makes
$\operatorname{poly}(m,1/\delta)$ calls to the single-dimensional
optimization procedure, and returns feasible caps
$\widehat q_j$ satisfying
\[
 \sum_j\left[\sig(V_j;\widehat q_j)
             -\frac{e}{e-1}\lambda_j\widehat q_j\right]
 \ge\mathcal R_\lambda-\delta S.
\]
\end{lemma}

\begin{proof}
To bound the discretization error,  
we normalize values, payments, and demand weights by
dividing them by $S$. All quantities in the following calculation
are expressed in normalized units, using the same notation;
scaling by a positive constant preserves the greedy choices.
For an integer $K\ge1$, apply the greedy algorithm to obtain caps
$\widehat q_j$, with objective
$\Phi_{\mathrm G}=\sum_j[\sig(V_j;\widehat q_j)
-\frac{e}{e-1}\lambda_j\widehat q_j]$.
As before, let $\widehat q_j^+:=\lceil K\widehat q_j^\star\rceil/K$
be the rounded optimal caps and write $\Phi_+$ for their objective.
Rounding upward does not decrease the utility terms, by
monotonicity of $\sig$, but increases each cap's demand penalty
by at most $\frac1K \cdot \frac{e}{e-1}\lambda_j$.
Subtracting this upper bound on the rounding loss from
the continuous optimum gives
\[
 \mathcal R_\lambda-\frac{e}{(e-1)K}\sum_{j=1}^m\lambda_j
  \le\Phi_+.
\]
Separately, the greedy comparison in
\eqref{eq:greedy-prefix-comparison}, with $C=1$ and $d=m$,
bounds the rounded objective in terms of the returned objective:
\[
    \Phi_+\le\left(1+\frac mK\right)\Phi_{\mathrm G}.
\]

Combining these two inequalities gives
\[
 \mathcal R_\lambda-\Phi_{\mathrm G}
 \le\frac{m\mathcal R_\lambda+\frac{e}{e-1}\sum_j\lambda_j}{K+m}
 \le\frac{m(1+e/(e-1))}{K+m}
  \le\frac{m(1+e/(e-1))}{K}
 \le\delta,
\]
where the second inequality uses $\mathcal R_\lambda\le\sum_j\E[V_j]\le1$, since payments and
demand penalties are nonnegative and values are normalized by $S$,
and $\sum_j\lambda_j\le m$, since demand weights $\lambda_j$ are bounded by
$S$ before normalization and thus by 1 afterward,
the last inequality is by the definition of $K$.

The additive error is thus $\delta$ in normalized units, 
or $\delta S$ in the original units.
Finally, the greedy algorithm uses at most
$m+K=O(m/\delta)$ function evaluations, with polynomial
running time as established for the general procedure, giving the stated
computational guarantees.
\end{proof}

\subsection{Proof of the demand-weight balancing lemma}
\label{app:non-iid-average-proof}

In this subsection, we prove \Cref{lem:non-iid-balanced}, restated below. 
We first bound the average utility minus the excess demand generated
by the demand-weight updates, then show that reducing how often
the resulting mechanisms run restores ex-ante feasibility while
preserving the claimed utility guarantee.

Recall the construction of $\indmech$ in
\Cref{sec:lambdachoice}: starting from $\lambda^1=0$, we compute
single-buyer mechanisms $A_i^t$ in each of $T$ rounds and update
the demand weights using their expected requests. The final
mechanism independently samples a round for each buyer and runs
the corresponding $A_i^t$, possibly making no offer to reduce
expected demand. We analyze these $T$ rounds to justify that
choice of how often the sampled mechanisms run.

Let $S:=\sum_{i,j}\E[V_{ij}]$. Divide values, payments, and
demand weights by $S$ for the purposes of this subsection; this preserves incentives and
multiplicative approximation ratios.
The update \eqref{eq:non-iid-price-update} becomes
\[
  \lambda_j^{t+1}
 :=\min\left\{1,\max\left\{0,
 \lambda_j^t+\eta(r_j^t-1)\right\}\right\}.
\]

Write $U^t:=\sum_iU(A_i^t)$ and retain the definitions
of expected requests $r_j^t$ from \Cref{sec:lambdachoice}. Write $\overline U:=T^{-1}\sum_tU^t$ and
$\overline r_j:=T^{-1}\sum_tr_j^t$ for the average total utility
and the average number of requests for item $j$. We will lower-bound
$\overline U-\sum_j(\overline r_j-1)^+$: the utility remaining
after subtracting each item's expected requests above its supply
of 1. This bound will let us reduce how often the mechanisms run
to restore ex-ante feasibility while preserving high utility.

To obtain this bound, we compare the demand-weight updates with a
fixed vector that assigns weight 1 to items with
$\overline r_j>1$ and zero to the others. Its demand-adjustment
term is exactly $-\sum_j(\overline r_j-1)^+$. Combining this
comparison with the utility inequality that holds in every round
gives the following guarantee for the averaged mechanisms.

\begin{restatable}{lemma}{noniidaverage}\label{lem:non-iid-average}
With step size $\eta=1/(n\sqrt T)$, after $T$
demand-weight update rounds,
\[
 \overline U-\sum_j(\overline r_j-1)^+
  \ge\frac{\OPT}{2+e/(e-1)}-\frac{nm}{\sqrt T}-n\delta.
\]
where $\delta\in(0,1]$ is the additive error in
each constructed single-buyer mechanism's demand-adjusted utility
guarantee.
\end{restatable}
\begin{proof}[Proof of \Cref{lem:non-iid-average}]
We follow the online gradient descent analysis of
\citet{Zinkevich03}, giving the details for our demand-weight updates.
Write $r^t=(r_1^t,\ldots,r_m^t)$ and
$\overline r=T^{-1}\sum_{t=1}^T r^t$, and let $\mathbf 1$ be the
all-1s vector. We compare the term $\lambda^t\cdot(\mathbf 1-r^t)$
with the corresponding term obtained by using one fixed weight vector
$\lambda^{\mathrm{fix}}\in[0,1]^m$ in every round, with an average
error that vanishes as $T$ grows. We leave $\lambda^{\mathrm{fix}}$ arbitrary
so that, after bounding this comparison, we can choose its weights
to count exactly the items' average excess demand. This vector is
used only in the proof and does not change the update procedure.
To obtain this comparison, we track how the squared distance
from $\lambda^{\mathrm{fix}}$ changes in each update:
\begin{align*}
 \|\lambda^{t+1}-\lambda^{\mathrm{fix}}\|_2^2
 &\le\|\lambda^t+\eta(r^t-\mathbf 1)-\lambda^{\mathrm{fix}}\|_2^2\\
 &=\|\lambda^t-\lambda^{\mathrm{fix}}\|_2^2
   -2\eta(\lambda^t-\lambda^{\mathrm{fix}})\cdot(\mathbf 1-r^t)
   +\eta^2\|r^t-\mathbf 1\|_2^2,
\end{align*}
where the inequality is because clipping $\lambda^{t+1}$ moves $\lambda^t+\eta(r^t-\mathbf 1)$ inside $[0,1]^m$ (if required),
and $\lambda_j^{\mathrm{fix}}$ also lies $\in[0,1]^m$, and therefore clipping can only reduce the distance
to $\lambda_j^{\mathrm{fix}}$.

Rearranging and dividing by $2\eta>0$ yields
\begin{equation}\label{eq:appendix-price-one-round}
 (\lambda^t-\lambda^{\mathrm{fix}})\cdot(\mathbf 1-r^t)
 \le \frac{\|\lambda^t-\lambda^{\mathrm{fix}}\|_2^2-\|\lambda^{t+1}-\lambda^{\mathrm{fix}}\|_2^2}{2\eta}
      +\frac{\eta}{2}\|r^t-\mathbf 1\|_2^2.
\end{equation}

We now sum inequality~\eqref{eq:appendix-price-one-round}
over rounds and divide by $T$:
\begin{align}
 \frac1T\sum_{t=1}^T\lambda^t\cdot(\mathbf 1-r^t)
 &=\lambda^{\mathrm{fix}}\cdot(\mathbf 1-\overline r)
   +\frac1T\sum_{t=1}^T(\lambda^t-\lambda^{\mathrm{fix}})\cdot(\mathbf 1-r^t)
   \notag \\
 &\le\lambda^{\mathrm{fix}}\cdot(\mathbf 1-\overline r)
   +\frac{1}{2\eta T}\sum_{t=1}^T
     \left(\|\lambda^t-\lambda^{\mathrm{fix}}\|_2^2
           -\|\lambda^{t+1}-\lambda^{\mathrm{fix}}\|_2^2\right)
   +\frac{\eta}{2T}\sum_{t=1}^T\|r^t-\mathbf1\|_2^2
   \notag\\
 &=\lambda^{\mathrm{fix}}\cdot(\mathbf 1-\overline r)
   +\frac{\|\lambda^1-\lambda^{\mathrm{fix}}\|_2^2
          -\|\lambda^{T+1}-\lambda^{\mathrm{fix}}\|_2^2}{2\eta T}
   +\frac{\eta}{2T}\sum_{t=1}^T\|r^t-\mathbf1\|_2^2
   \notag\\
 &\le\lambda^{\mathrm{fix}}\cdot(\mathbf 1-\overline r)
   +\frac{\sum_{j=1}^m(\lambda_j^{\mathrm{fix}})^2}{2\eta T}
   +\frac{\eta}{2T}\,Tmn^2
   \notag\\
 &\le\lambda^{\mathrm{fix}}\cdot(\mathbf 1-\overline r)+\frac{nm}{\sqrt T},
 \label{eq:appendix-price-comparator}
\end{align}
where the first equality adds and subtracts the fixed-weight term and
uses $\overline r=T^{-1}\sum_t r^t$, the first inequality applies
\eqref{eq:appendix-price-one-round}, the second equality
telescopes terms, the second inequality uses
$\lambda^1=0$, drops the nonpositive final-distance term and
uses $|r_j^t-1|\le n$ (since $r_j^t$ is the sum of $n$ request
probabilities), and the last inequality uses
$0\le\lambda_j^{\mathrm{fix}}\le1$ and $\eta=1/(n\sqrt T)$.

Since this bound holds for every $\lambda^{\mathrm{fix}}\in[0,1]^m$, choose
$\lambda_j^{\mathrm{fix}}=1$ if $\overline r_j>1$ and $\lambda_j^{\mathrm{fix}}=0$ otherwise. For this
choice, $\lambda^{\mathrm{fix}}\cdot(\mathbf 1-\overline r)=-\sum_j(\overline r_j-1)^+$, so
inequality~\eqref{eq:appendix-price-comparator} turns into
\begin{equation}\label{eq:appendix-price-excess}
 \frac1T\sum_{t=1}^T\lambda^t\cdot(\mathbf 1-r^t)
 \le -\sum_j(\overline r_j-1)^+ +\frac{nm}{\sqrt T}.
\end{equation}

Computing approximate caps with \Cref{lem:greedy-demand-adjusted}
and using them to design the mechanisms prescribed by
\Cref{lem:non-iid-local} incurs at most $\delta$ additive error
per buyer in normalized units. Thus, the total error across
the $n$ buyers is at most $n\delta$, giving in every round
\[
 U^t+\lambda^t\cdot(\mathbf1-r^t)
 \ge\frac{\OPT}{2+e/(e-1)}-n\delta.
\]
 Averaging this inequality over the $T$ rounds and 
 using \eqref{eq:appendix-price-excess} proves the lemma.
\[
 \frac{\OPT}{2+e/(e-1)}-n\delta
 \le\overline U
      +\frac1T\sum_{t=1}^T\lambda^t\cdot(\mathbf 1-r^t)
 \le\overline U-\sum_j(\overline r_j-1)^+
      +\frac{nm}{\sqrt T}.
\]
\end{proof}

We are now ready to prove \Cref{lem:non-iid-balanced}.

\noniidbalanced*

\begin{proof}[Proof of \Cref{lem:non-iid-balanced}]
It suffices to consider
$0<\varepsilon<1/(2+e/(e-1))$, since otherwise the Lemma's
utility lower bound is nonpositive (and is trivially true). 
We maintain the normalized units used in this subsection.
Choose
\[
  T:=\left\lceil\frac{4n^4m^4}{\varepsilon^2}\right\rceil,
 \qquad \delta\le\frac{\varepsilon}{2n^2m}.
\]
By the definition of $\indmech$, its total expected
number of requests for each item $j$ is $\frac{\overline r_j}{\max\{1,\max_k\overline r_k\}}\le1$
for every $j$, making mechanism $\indmech$ ex-ante feasible.

We now bound the utility lost by reducing participation:
\begin{align}
U(\indmech)
 &=\frac{\overline U}{\max\{1,\max_j\overline r_j\}}
       \notag & \text{definition of $\indmech$}\\
 &\ge\overline U-\max_j(\overline r_j-1)^+
       \notag & \text{$\overline U\le1$, $1-1/z\le z-1$ for $z\ge1$}\\
 &\ge\overline U-\sum_j(\overline r_j-1)^+
       \notag & \text{maximum bounded by sum}\\
 &\ge\frac{\OPT}{2+e/(e-1)}-\frac{nm}{\sqrt T}-n\delta.
       & \text{\Cref{lem:non-iid-average}}
 \label{eq:non-iid-request-scaling}
\end{align}
The first inequality uses $\overline U\le1$ which is because utility is
upper bounded by welfare, and
expected welfare is at most $\sum_{i,j}\E[V_{ij}]=1$, even without
enforcing a matching allocation.
The final guarantee follows:
\begin{align*}
 U(\indmech)
 &\ge\frac{\OPT}{2+e/(e-1)}-\frac{nm}{\sqrt T}-n\delta
   &\text{\eqref{eq:non-iid-request-scaling}}\\
 &\ge\frac{\OPT}{2+e/(e-1)}
       -\frac{\varepsilon}{2nm}-\frac{\varepsilon}{2nm}
   &\text{choices of $T$ and $\delta$}\\
 &\ge\frac{\OPT}{2+e/(e-1)}-\varepsilon\OPT
   &\text{$\OPT\ge1/(nm)$}\\
 &=\left(\frac{1}{2+e/(e-1)}-\varepsilon\right)\OPT,
\end{align*}
where $\OPT\ge1/(nm)$ because allocating the item for free to the buyer in the buyer-item pair with the largest expected value is feasible and yields utility at least $1/(nm)$, since $\sum_{i,j}\E[V_{ij}]=1$.
Multiplying values
and payments by the normalizing factor $S$ restores the original units without changing
request probabilities or incentive properties.
\end{proof}

\subsection{Computing request and contention probabilities}
\label{app:non-iid-probabilities}

This subsection explains how to compute the request
probabilities $p_{ij}$ and selection probabilities $c_{ij}$ required
to implement the contention-resolution scheme in \Cref{sec:hetservice}.
These probabilities are computed before buyers report their values
or draw their rounds. A buyer's request
probabilities depend only on that buyer's value distributions and
single-buyer mechanisms; contention probabilities use the final
request probabilities of all buyers for the same item.

For each round $t$, let $p_{ij}^t:=\eaqj(A_i^t)$ be the
probability that buyer $i$ requests item $j$ under mechanism $A_i^t$.
To compute this for a fixed price vector, you condition on item $j$ yielding positive
surplus $y$. Independence makes the probability that $j$ is chosen
the product of the probabilities that each other item's surplus
is below $y$. Averaging over the conditioned
item $j$'s surplus $y$, assigning zero probability when it is nonpositive,
gives the request probability for a fixed price. This idea can be applied to the pricing sub-mechanism. 
For the free-lottery sub-mechanism we can use (weighted) values instead of surpluses,
multiplying the probability of choosing item $j$ by $\ell_j(h)$. Taking into account the mixing probabilities of the sub-mechanisms, we can thus compute $p_{ij}^t$. 
We assume these probability evaluations can be performed efficiently from the value distributions.

The final request probabilities are
\[
 p_{ij}=\frac{1}{T\max\{1,\max_k\overline r_k\}}
          \sum_{t=1}^T p_{ij}^t.
\]
The formula averages over all rounds and includes the probability
of making no offer; it does not condition on the drawn round.

To compute the probability $c_{ij}$ that the
contention rule selects buyer $i$ when they request item $j$,
we include buyer $i$ among the requesters for item $j$ and let every other buyer $k$ request
independently with probability $p_{kj}$. Then $c_{ij}$ is buyer
$i$'s probability of selection under the contention rule.
We also define $c_{ij}$ when $p_{ij}=0$, keeping the rule's
probability parameters unchanged, to preserve incentives for requests
that can arise through misreports. 

If $i$ is the only buyer requesting item $j$, the
contention rule in \Cref{sec:hetservice} selects them with
probability 1; if there are other requests, forming a requester
set $A$ of size $r\ge2$, their selection probability is
\begin{equation}\label{eq:non-iid-rationing}
 \frac1r\left(
  1+\frac{\sum_{k\in A\setminus\{i\}}p_{kj}}{r-1}-p_{ij}
 \right).
\end{equation}
Let $K:=|A\setminus\{i\}|$
count the other requesters. Averaging the contention rule with
$r=k+1$ on the event $K=k$ gives
\[
\begin{aligned}
 c_{ij}
 &=\Prb(K=0)+\sum_{k=1}^{n-1}
   \E\!\left[\frac1{k+1}\left(
       1+\frac{\sum_{\ell\in A\setminus\{i\}}p_{\ell j}}{k}-p_{ij}
       \right)
              \mathbf1_{\{K=k\}}\right]\\
 &=\Prb(K=0)+\sum_{k=1}^{n-1}
 \frac{(1-p_{ij})\Prb(K=k)
             +\frac1k\E\!\left[
                \left(\sum_{\ell\in A\setminus\{i\}}p_{\ell j}\right)
                \mathbf1_{\{K=k\}}\right]}{k+1}.
\end{aligned}
\]
For each pair $(i,j)$, we compute $c_{ij}$ in $O(n^2)$ 
operations by the following procedure. We compute
$a_k:=\Prb(K=k)$ and
$b_k:=\E[(\sum_{\ell\in A\setminus\{i\}}p_{\ell j})\mathbf1_{\{K=k\}}]$
by processing the other buyers one at a time.
Start with $a_0=1$, $b_0=0$, and zero elsewhere. When the next
buyer $\ell$ has final request probability $p=p_{\ell j}$, update
the arrays to
\[
 a'_k=(1-p)a_k+pa_{k-1},\qquad
 b'_k=(1-p)b_k+pb_{k-1}+p^2a_{k-1},
\]
where entries with negative indices are zero. The first
update separates the cases in which this buyer does or does not
request $j$. In the second, a request also adds $p$ to the sum
of the requesters' probabilities; this
event has probability $p$, producing the term $p^2a_{k-1}$.
After processing all buyers other than $i$, substitute the resulting
$a_k$ and $b_k$ into the displayed formula for $c_{ij}$.

\end{document}

%% file: paper-macros.tex
\newtheorem{corollary}{Corollary}
\newtheorem{observation}{Observation}
\crefname{observation}{observation}{observations}
\Crefname{observation}{Observation}{Observations}
\newcommand{\mfree}{\mathcal M^{\mathrm{free}}}
\newcommand{\mpp}{\mathcal{M}^{UD}}
\newcommand{\jpp}{\mathcal{M}^{1}_{j}}
\newcommand{\E}{\mathbb E}
\newcommand{\Prb}{\Pr}
\newcommand{\Vvec}{\mathbf V}
\newcommand{\Wvec}{\mathbf W}
\newcommand{\Zvec}{\mathbf Z}
\newcommand{\Vprofile}{\mathcal V}
\newcommand{\Rplus}{\mathbb R_+}
\newcommand{\SingleOPT}{\mathrm{OPT}^{1,m}}
\newcommand{\OPT}{\mathrm{OPT}^{n,m}}

\newcommand{\UP}{U(\mpp)}
\newcommand{\UF}{U(\mfree)}
\newcommand{\utilsingleiidmech}{U(\singleiidmech)}
\newcommand{\M}{\mathcal{M}}
\renewcommand{\L}{\mathcal{L}}
\newcommand{\G}{G_{\lambda}}
\newcommand{\eaq}{q^{\mathrm{alloc}}}
\newcommand{\eaqj}{\eaq_j}
\newcommand{\sig}{\mathrm{OPT}^{1,1}}
\newcommand{\singlehetmech}{\mathcal M_{\mathrm{het}}^{1,m}}
\newcommand{\singleiidmech}{\mathcal M_{\mathrm{iid}}^{1,m}(Q)}
\newcommand{\iidmech}{\mathcal M_{\mathrm{iid}}^{n,m}}
\newcommand{\hetmech}{\mathcal M_{\mathrm{het}}^{n,m}}
\newcommand{\indmech}{\overline{\mathcal M}_{\mathrm{het}}^{n,m}}
\newtheorem{theorem}{Theorem}
\newtheorem{lemma}{Lemma}
\numberwithin{equation}{section}

%% file: masterbib.bib
@INPROCEEDINGS{Alaei11,
author={Alaei, S.},
booktitle={Foundations of Computer Science (FOCS), 2011 IEEE 52nd Annual Symposium on},
title={{Bayesian Combinatorial Auctions: Expanding Single Buyer Mechanisms to Many Buyers}},
year={2011},
pages={512-521},
doi={10.1109/FOCS.2011.90},
ISSN={0272-5428},
month={Oct},}

@inproceedings{AFHH13,
  author    = {Alaei, Saeed and Fu, Hu and Haghpanah, Nima
               and Hartline, Jason D.},
  title     = {{The Simple Economics of Approximately Optimal Auctions}},
  booktitle = {Proceedings of the 54th Annual IEEE Symposium
               on Foundations of Computer Science (FOCS)},
  pages     = {628--637},
  year      = {2013},
  publisher = {IEEE},
  doi       = {10.1109/FOCS.2013.73},
  url       = {https://arxiv.org/abs/1206.3541}
}

@inproceedings{BILW,
	Author = {Babaioff, Moshe and Immorlica, Nicole and Lucier, Brendan and Weinberg, S. Matthew},
	Booktitle = {Foundations of Computer Science (FOCS), 2014 IEEE 55th Annual Symposium on},
	Organization = {IEEE},
	Pages = {21--30},
	Title = {{A Simple and Approximately Optimal Mechanism for an
                   Additive Buyer}},
	Year = {2014}}

@inproceedings{BOT25,
author = {Berzack, Batya and Oshman, Rotem and Talgam-Cohen, Inbal},
title = {{Dynamic Rental Games with Stagewise Individual Rationality}},
year = {2025},
isbn = {9798400719431},
publisher = {Association for Computing Machinery},
address = {New York, NY, USA},
url = {https://doi.org/10.1145/3736252.3742624},
doi = {10.1145/3736252.3742624},
booktitle = {Proceedings of the 26th ACM Conference on Economics and Computation},
pages = {785},
numpages = {1},
location = {Stanford University, Stanford, CA, USA},
series = {EC '25}
}

@misc{BOT26,
      title={{Optimal Auction Design for Constrained Buyers}}, 
      author={Batya Berzack and Rotem Oshman and Inbal Talgam-Cohen},
      year={2026},
      eprint={2606.30776},
      archivePrefix={arXiv},
      primaryClass={cs.GT},
      url={https://arxiv.org/abs/2606.30776}, 
}

@misc{BoydSubgradient,
  author = {Boyd, S.},
  title = {{Subgradient Methods}},
  howpublished = {EE364b lecture notes, Stanford University},
  year = {2014},
  note = {With help from J. Park},
  url = {https://web.stanford.edu/class/ee364b/lectures/subgrad_method_notes.pdf}
}

@article{BravermanCK16,
  title={{Optimal Provision-After-Wait in Healthcare}},
  author={Braverman, Mark and Chen, Jing and Kannan, Sampath},
  journal={Mathematics of Operations Research},
  volume={41},
  number={1},
  pages={352--376},
  year={2016},
  publisher={INFORMS}
}

@inproceedings{CDW,
 author = {Cai, Yang and Devanur, Nikhil R. and Weinberg, S. Matthew},
 title = {{A Duality Based Unified Approach to Bayesian Mechanism Design}},
 booktitle = {Proceedings of the Forty-eighth Annual ACM Symposium on Theory of Computing},
 series = {STOC '16},
 year = {2016},
 isbn = {978-1-4503-4132-5},
 location = {Cambridge, MA, USA},
 pages = {926--939},
 numpages = {14},
 url = {http://doi.acm.org/10.1145/2897518.2897645},
 doi = {10.1145/2897518.2897645},
 acmid = {2897645},
 publisher = {ACM},
 address = {New York, NY, USA},
}

@inproceedings{CZ17,
 author = {Cai, Yang and Zhao, Mingfei},
 title = {{Simple Mechanisms for Subadditive Buyers via Duality}},
 booktitle = {Proceedings of the 49th Annual ACM SIGACT Symposium on Theory of Computing},
 series = {STOC 2017},
 year = {2017},
 isbn = {978-1-4503-4528-6},
 location = {Montreal, Canada},
 pages = {170--183},
 numpages = {14},
 url = {http://doi.acm.org/10.1145/3055399.3055465},
 doi = {10.1145/3055399.3055465},
 acmid = {3055465},
 publisher = {ACM},
 address = {New York, NY, USA},
}

@inproceedings{CHMS,
	Author = {Chawla, Shuchi and Hartline, Jason D and Malec, David L and Sivan, Balasubramanian},
	Booktitle = {Proceedings of the forty-second ACM symposium on Theory of computing},
	Organization = {ACM},
	Pages = {311--320},
	Title = {{Multi-Parameter Mechanism Design and Sequential Posted Pricing}},
	url = {http://dl.acm.org/citation.cfm?id=1806733},
	Year = {2010}}

@inproceedings{CMM11,
  title={Bayesian mechanism design for budget-constrained agents},
  author={Chawla, Shuchi and Malec, David L and Malekian, Azarakhsh},
  booktitle={Proceedings of the 12th ACM conference on Electronic commerce},
  pages={253--262},
  year={2011}
}

@article{CMS15,
title = {{The Power of Randomness in Bayesian Optimal Mechanism Design}},
journal = {Games and Economic Behavior},
volume = {91},
number = {},
pages = {297 - 317},
year = {2015},
note = {},
issn = {0899-8256},
doi = {http://dx.doi.org/10.1016/j.geb.2012.08.010},
url = {http://www.sciencedirect.com/science/article/pii/S0899825612001303},
author = {Shuchi Chawla and David Malec and Balasubramanian Sivan}
}

@inproceedings{CM,
	Acmid = {2940756},
	Address = {New York, NY, USA},
	Author = {Chawla, Shuchi and Miller, J. Benjamin},
	Booktitle = {Proceedings of the 2016 ACM Conference on Economics and Computation},
	Doi = {10.1145/2940716.2940756},
	Isbn = {978-1-4503-3936-0},
	Location = {Maastricht, The Netherlands},
	Numpages = {18},
	Pages = {579--596},
	Publisher = {ACM},
	Series = {EC '16},
	Title = {{Mechanism Design for Subadditive Agents via an Ex Ante Relaxation}},
	Url = {http://doi.acm.org/10.1145/2940716.2940756},
	Year = {2016}}

@inproceedings{chawla2023buy,
  title={{Buy-Many Mechanisms for Many Unit-Demand Buyers}},
  author={Chawla, Shuchi and Rezvan, Rojin and Teng, Yifeng and Tzamos, Christos},
  booktitle={International Conference on Web and Internet Economics},
  pages={21--38},
  year={2023},
  organization={Springer}
}

@inproceedings{EFFGK,
 author = {Eden, Alon and Feldman, Michal and Fiat, Amos and Goldner, Kira and Karlin, Anna R.},
 title = {{Combinatorial Auctions with Interdependent Valuations: SOS to the Rescue}},
 booktitle = {Proceedings of the 2019 ACM Conference on Economics and Computation},
 series = {EC '19},
 year = {2019},
 isbn = {978-1-4503-6792-9},
 location = {Phoenix, AZ, USA},
 numpages = {2},
 url = {https://arxiv.org/abs/1903.08384},
 doi = {10.1145/3328526.3329759},
 publisher = {ACM},
 address = {Phoenix, AZ, USA},
}

@inproceedings{EFGMMfocs,
  title={{Constant Approximation for Private Interdependent Valuations}},
  author={Eden, Alon and Feldman, Michal and Goldner, Kira and Mauras, Simon and Mohan, Divyarthi},
  booktitle={2023 IEEE 64th Annual Symposium on Foundations of Computer Science (FOCS)},
  pages={148--163},
  year={2023},
  organization={IEEE}
}

@inproceedings{EGKT26,
  title={{From Welfare to Utility: Generalized Objectives in Budget-Feasible Procurement}},
  author={Eden, Alon and Goldner, Kira and Kerner, Eldar and Tsilivis, Thodoris},
  booktitle={Forty-third International Conference on Machine Learning},
  year={2026}
}

@inproceedings{EGZ,
  title={{Private Interdependent Valuations}},
  author={Eden, Alon and Goldner, Kira and Zheng, Shuran},
  booktitle={Proceedings of the 2022 Annual ACM-SIAM Symposium on Discrete Algorithms (SODA)},
  pages={2920--2939},
  year={2022},
  organization={SIAM}
}

@inproceedings{EssaidiGW24,
 author = {Meryem Essaidi and
               Kira Goldner and
               S. Matthew Weinberg},
 title = {{To Regulate or Not to Regulate: Using Revenue Maximization Tools to Maximize Consumer Utility}},
 booktitle = {Proceedings of the 17th International Symposium on Algorithmic Game Theory (SAGT)},
editor="Sch{\"a}fer, Guido
and Ventre, Carmine",
address="Cham",
pages="315--332",
isbn="978-3-031-71033-9",
 url       = {https://arxiv.org/abs/2002.06326},
 year = {2024}
}

@inproceedings{ezra2024optimal,
  title={{Multi-Parameter Mechanisms for Consumer Surplus Maximization}},
  author={Ezra, Tomer and Schoepflin, Daniel and Shaulker, Ariel},
  booktitle={Proceedings of the 57th Annual ACM Symposium on Theory of Computing},
  pages={483--494},
  year={2025}
}

@article{fotakis2016efficient,
  title={{Efficient Money Burning in General Domains}},
  author={Fotakis, Dimitris and Tsipras, Dimitris and Tzamos, Christos and Zampetakis, Emmanouil},
  journal={Theory of Computing Systems},
  volume={59},
  pages={619--640},
  year={2016},
  publisher={Springer}
}

@inproceedings{FV06,
author = { Feige, Uriel and Vondrak, Jan },
booktitle = { 2006 47th Annual IEEE Conference on Foundations of Computer Science },
title = {{ Approximation Algorithms for Allocation Problems: Improving the Factor of 1 - 1/e }},
year = {2006},
volume = {},
ISSN = {0272-5428},
pages = {667-676},
doi = {10.1109/FOCS.2006.14},
publisher = {IEEE Computer Society},
address = {Los Alamitos, CA, USA},
month =Oct}

@inproceedings{GH24,
  title={{Combinatorial Pen Testing (Or Consumer Surplus of Deferred-Acceptance Auctions)}},
  author={Ganesh, Aadityan and Hartline, Jason},
  booktitle={16th Innovations in Theoretical Computer Science Conference (ITCS 2025)},
  pages={52--1},
  year={2025},
  organization={Schloss Dagstuhl--Leibniz-Zentrum f{\"u}r Informatik}
}

@inproceedings{goldnerlundy25,
  title={{Multidimensional Bayesian Utility Maximization: Tight Approximations to Welfare}},
  author={Goldner, Kira and Lundy, Taylor},
  booktitle={Proceedings of the Thirty-ninth Annual Conference on Neural Information Processing Systems},
  year={2025}
}

@inproceedings{HN13,
 author = {Hart, Sergiu and Nisan, Noam},
 title = {{The Menu-Size Complexity of Auctions}},
 booktitle = {Proceedings of the Fourteenth ACM Conference on Electronic Commerce},
 series = {EC '13},
 year = {2013},
 isbn = {978-1-4503-1962-1},
 location = {Philadelphia, Pennsylvania, USA},
 pages = {565--566},
 numpages = {2},
 url = {http://doi.acm.org/10.1145/2482540.2482544},
 doi = {10.1145/2482540.2482544},
 acmid = {2482544},
 publisher = {ACM},
 address = {New York, NY, USA},
}

@article{HartlineBook,
  title={{Mechanism Design and Approximation}},
  author={Hartline, Jason D.},
  journal={Book draft. October},
  volume={122},
  url = {http://jasonhartline.com/MDnA/},
  year={2013}
}

@inproceedings{HartlineR08,
 author = {Hartline, Jason D. and Roughgarden, Tim},
 title = {{Optimal Mechanism Design and Money Burning}},
 booktitle = {STOC '08: Proceedings of the 40th annual ACM symposium on Theory of computing},
 year = {2008},
 isbn = {978-1-60558-047-0},
 pages = {75--84},
 location = {Victoria, British Columbia, Canada},
 doi = {http://doi.acm.org/10.1145/1374376.1374390},
 publisher = {ACM},
 address = {New York, NY, USA},
 }

@article{Heflin,
  title={{Local Control, Discretion, and Administrative Burden: SNAP Interview Waivers and Caseloads During the COVID-19 Pandemic}},
  author={Heflin, Colleen  and Fannin, William Clay and Lopoo, Leonard},
  journal={American Review of Public Administration},
  year={2023}
}

@article{KashaevSantiago23,
  author = {Kashaev, D. and Santiago, R.},
  title = {{A Simple Optimal Contention Resolution Scheme for Uniform Matroids}},
  journal = {Theoretical Computer Science},
  volume = {940},
  pages = {81--96},
  year = {2023},
  url = {https://arxiv.org/abs/2105.11992}
}

@inproceedings{KleinbergW12,
  author    = {Robert Kleinberg and
               S. Matthew Weinberg},
  editor    = {Howard J. Karloff and
               Toniann Pitassi},
  title     = {{Matroid Prophet Inequalities}},
  booktitle = {Proceedings of the 44th Symposium on Theory of Computing Conference,
               {STOC} 2012, New York, NY, USA, May 19 - 22, 2012},
  pages     = {123--136},
  publisher = {{ACM}},
  year      = {2012},
  url       = {https://doi.org/10.1145/2213977.2213991},
  doi       = {10.1145/2213977.2213991},
  bibsource = {dblp computer science bibliography, https://dblp.org}
}

@inproceedings{lundy2024pay,
  title={Pay to (Not) Play: Monetizing Impatience in Mobile Games},
  author={Lundy, Taylor and Raman, Narun and Fu, Hu and Leyton-Brown, Kevin},
  booktitle={Proceedings of the AAAI Conference on Artificial Intelligence},
  volume={38},
  number={9},
  pages={9856--9864},
  year={2024}
}

@article{Samuel-Cahn,
  title={{Comparison of Threshold Stop Rules and Maximum for Independent Nonnegative Random Variables}},
  author={Samuel-Cahn, Ester},
  journal={The Annals of Probability},
  pages={1213--1216},
  year={1984},
  publisher={JSTOR}
}

@inproceedings{Yao15,
	Author = {Yao, Andrew Chi-Chih},
	Booktitle = {Proceedings of the Twenty-Sixth Annual ACM-SIAM Symposium on Discrete Algorithms},
	Organization = {SIAM},
	Pages = {92--109},
	Title = {{An $n$-to-1 Bidder Reduction for Multi-Item Auctions
                   and Its Applications}},
	Year = {2015},
	doi = {10.1137/1.9781611973730.8},
        url =  {http://arxiv.org/abs/1406.3278}}

@inproceedings{Zinkevich03,
  author    = {Martin Zinkevich},
  title     = {Online Convex Programming and Generalized
               Infinitesimal Gradient Ascent},
  booktitle = {Proceedings of the Twentieth International
               Conference on Machine Learning (ICML)},
  pages     = {928--936},
  year      = {2003}
}

@inproceedings{lundy2019allocation,
  title={Allocation for social good: auditing mechanisms for utility maximization},
  author={Lundy, Taylor and Wei, Alexander and Fu, Hu and Kominers, Scott Duke and Leyton-Brown, Kevin},
  booktitle={Proceedings of the 2019 ACM Conference on Economics and Computation},
  pages={785--803},
  year={2019}
}

@article{patel2022costly,
  title={Costly Verification and Money Burning},
  author={Patel, Rohit and Urgun, Can},
  year={2022},
  publisher={Working Paper}
}

@misc{ezra2026moneyburning,
  author = {Ezra, Tomer},
  title = {{Money Burning Mechanism Design: From Welfare to Surplus}},
  year = {2026},
  eprint = {2608.01693},
  archivePrefix = {arXiv},
  primaryClass = {cs.GT},
  note = {Version 2, September 24, 2026},
  url = {https://arxiv.org/abs/2608.01693v2}
}

@inproceedings{GMT26,
  author = {Goldner, Kira and Mohan, Divyarthi and Tsilivis, Thodoris},
  title = {{Knowing Who, Not How Much: Learning-Augmented Mechanisms
            for Consumer Utility Maximization}},
  booktitle = {Forty-third International Conference on Machine Learning},
  year = {2026},
  url = {https://arxiv.org/abs/2607.00175}
}

@misc{BMMP24,
  author = {Bhawalkar, Kshipra and Mertzanidis, Marios and
            Mohan, Divyarthi and Psomas, Alexandros},
  title = {{Mechanism Design via the Interim Relaxation}},
  year = {2024},
  eprint = {2407.12699},
  archivePrefix = {arXiv},
  primaryClass = {cs.GT},
  url = {https://arxiv.org/abs/2407.12699}
}
